%% file: main.tex
\documentclass[a4paper,11pt]{article}
\input{packages}

\input{commands.tex}

\begin{document}

\input{abstract}

\input{intro}

\input{related_work}

\input{hardness}

\input{lower_bounds}

\input{out_trees}

\input{star}

\input{tree_instances}

\input{conclusion}

\bibliographystyle{abbrv}
\bibliography{references.bib}

\appendix
\input{Appendix/single_source}

\input{Appendix/out_tree_hardness}

\end{document}

%% file: packages.tex
\usepackage[left=1in, right=1in, top=1in, bottom=1in]{geometry}
\usepackage{wrapfig}
\usepackage{amsmath, amsthm, amssymb}
\usepackage{pdfpages} 
\usepackage{enumitem}
\usepackage[linesnumbered,ruled,vlined]{algorithm2e}
\usepackage{booktabs}
\usepackage{comment}
\usepackage[parfill]{parskip}
\usepackage{graphicx}
\usepackage{epsfig}
\usepackage{bm}
\usepackage{cancel}
\usepackage{wrapfig}
\usepackage{array}
\usepackage{caption}
\usepackage{placeins}
\usepackage{url}
\usepackage{hyperref}
\hypersetup{hidelinks}
\usepackage[bottom]{footmisc}
\usepackage{verbatim}
\usepackage{color,soul,xcolor}

\usepackage{enumitem}
\usepackage{amsmath} 
\usepackage{amsthm} 
\usepackage{amsfonts} 
\usepackage{mathtools}
\DeclarePairedDelimiter{\ceil}{\lceil}{\rceil}
\usepackage{fullpage}
\usepackage{parskip} 
\usepackage{graphicx} 
\usepackage{placeins} 
\graphicspath{ {./images/} }
\usepackage{caption}
\usepackage{subcaption}
\usepackage[T1]{fontenc}

\usepackage{thm-restate}

\makeatletter
\renewenvironment{proof}[1][\normalfont\bfseries\proofname]{\par
  \pushQED{\qed}%
  \normalfont \topsep6\p@\@plus6\p@\relax
  \trivlist
  \item\relax
        {\itshape
    #1\@addpunct{.}}\hspace\labelsep\ignorespaces
}{%
  \popQED\endtrivlist\@endpefalse
}

\def\thm@space@setup{%
  \thm@preskip=0\parskip \thm@postskip=0pt
}
\makeatother

%% file: commands.tex
\usepackage[colorinlistoftodos,prependcaption,textsize=scriptsize]{todonotes}

\newcommand{\mnew}[1]{{\color{black}{#1}}}

\newcommand\argmax{\mbox{argmax}}

\newcommand\W{\mathcal{W}}
\renewcommand\S{\mathcal{S}}

\newcommand\T{\mathcal{T}}
\newcommand\I{\mathcal{I}}
\newcommand\K{\mathcal{K}}

\newcommand\C{\mathcal{C}}

\newcommand\A{\mathcal{A}}
\newcommand\cl{\mathtt{cl}}

\newcommand\LB{\mbox{LB}}
\renewcommand\P{\mathcal{P}}
\newcommand\mmspp{$\mathtt{min}$-$\mathtt{degree}$-$\mathtt{SPP}$}
\newcommand\hs{$\mathtt{Hitting}$-$\mathtt{set}$}

\graphicspath{ {./images/} }
\newtheorem{theorem}{Theorem}
\newtheorem{prop}[theorem]{Proposition}
\newtheorem{corollary}[theorem]{Corollary}
\newtheorem{lemma}[theorem]{Lemma}
\newtheorem{definition}[theorem]{Definition}

\newtheorem{open}[theorem]{Open question}

%% file: abstract.tex
\title{Fast Combinatorial Algorithms for Efficient Sortation}

\author{Madison Van Dyk\thanks{This work was partially supported by the NSERC Discovery Grant Program, grant number RGPIN-03956-2017 and by an Amazon Post-Internship Fellowship. The work presented here does not relate to the authors' positions at Amazon} \and Kim Klause\thanks{Faculty of Mathematics and Computer Science, University of Bremen} \and Jochen Koenemann$^*$\thanks{Modeling and Optimization, Amazon.com} \and Nicole Megow$^\dagger$}
\date{}
\maketitle

\begin{abstract}
Modern parcel logistic networks are designed to ship demand between given origin, destination pairs of nodes in an underlying directed network. Efficiency dictates that volume needs to be consolidated at intermediate nodes in typical hub-and-spoke fashion. In practice, such consolidation requires parcel sortation. In this work, we propose a mathematical model for the physical requirements, and limitations of parcel sortation. We then show that it is NP-hard to determine whether a feasible sortation plan exists. We discuss several settings, where \mbox{(near-)}feasibility of a given sortation instance can be determined efficiently. The algorithms we propose are fast and build on combinatorial {\em witness set} type lower bounds that are reminiscent and extend those used in earlier work on degree-bounded spanning trees and arborescences. 
\end{abstract}

%% file: intro.tex
\section{Introduction}

In modern parcel logistics operations, one broadly faces the problem of finding an {\em optimal} way to ship a given input demand between source-sink node pairs within an underlying fulfillment network, typically represented as a directed graph~${D = (N, A)}$. The input demand is given as a collection of {\em commodities}~$\{(s_k,t_k)\}_{k \in \K}$, consisting of pairs of source and sink nodes in $N$. In this work, we assume each commodity $k$ comes with a directed $s_k,t_k$-path, $P_k$, along which all packages associated with commodity $k$ must travel. This assumption is often made in practical applications, as discussed in \cite{LK+23}.  

\begin{wrapfigure}{r}{.38 \textwidth}
\centering
 \includegraphics[width=.3\textwidth]{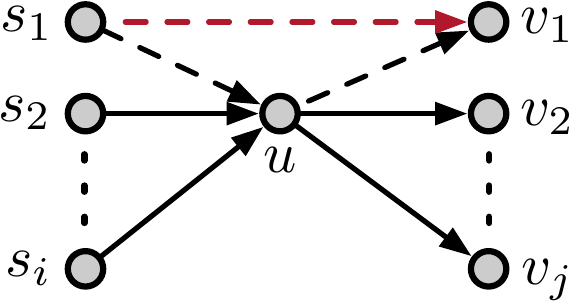}
 \caption{Sortation introduction}\label{fig:sort-intro}
\end{wrapfigure}

This paper focuses on {\em parcel sortation}, an aspect of routing that has previously been left unaddressed from a theory perspective. In the simplest setting, each package that travels from~$s_k$ to~$t_k$ via some internal node $u$ must be {\em sorted} at $u$ for its subsequent downstream node. In Figure \ref{fig:sort-intro}, packages arrive at node $u$ from nodes $s_1, \ldots s_i$ and travel on to downstream nodes $v_1, \ldots, v_j$. This requires sortation at node $u$ to sub-divide the stream of incoming parcels between the $j$ possible next stops. We refer to the physical device tasked with packages destined for one specific downstream node as a {\em sort point}. If the vast majority of the traffic on arc $(u,v_1)$ arrives at $u$ from $s_1$, then we could sort $s_1,v_1$ volume at $s_1$ to $v_1$ instead of to $u$. In practice this entails {\em containerizing} $s_1, v_1$ volume at $s_1$. These containers are not opened and sorted at intermediate node $u$, instead, they are {\em cross-docked}, an operation that is significantly faster and less costly than the process of sorting all individual parcels \cite{CrossDocking,SOAP}. 

Figure \ref{fig:sort-intro} illustrates the  containerization of $s_1,v_1$ volume at $s_1$ by replacing the two dashed arcs $(s_1,u)$ and $(u,v_1)$ by one red dashed arc $(s_1,v_1)$. Sort points are required at each node, for each outgoing arc. Hence, the arc replacement operation in Figure \ref{fig:sort-intro} reduces the number of required sort points at node $u$, but increases the number of sort points required at $s_1$ by 1. \mnew{Note that sort point capacity is often limited, and determining the maximum number of sort points required at any warehouse to route all commodities is important in both short- and long-term planning.} 

\textbf{A formal model for sortation.} Let $\K$ be a set of commodities, where each commodity $k \in \mathcal{K}$ has a source $s_k$, sink $t_k$, and a designated directed path $P_k$ in $D$. Let the \emph{transitive closure} of digraph $D$, denoted $\mathtt{cl}(D)$, be the graph obtained by introducing short-cut arcs $(i,j)$ whenever $j$ can be reached from $i$ through a directed path in $D$. An example is shown in Figure~\ref{fig:cl_example}. 
\begin{figure}[h!]
	\begin{center}
		\scalebox{0.36}{
			\includegraphics{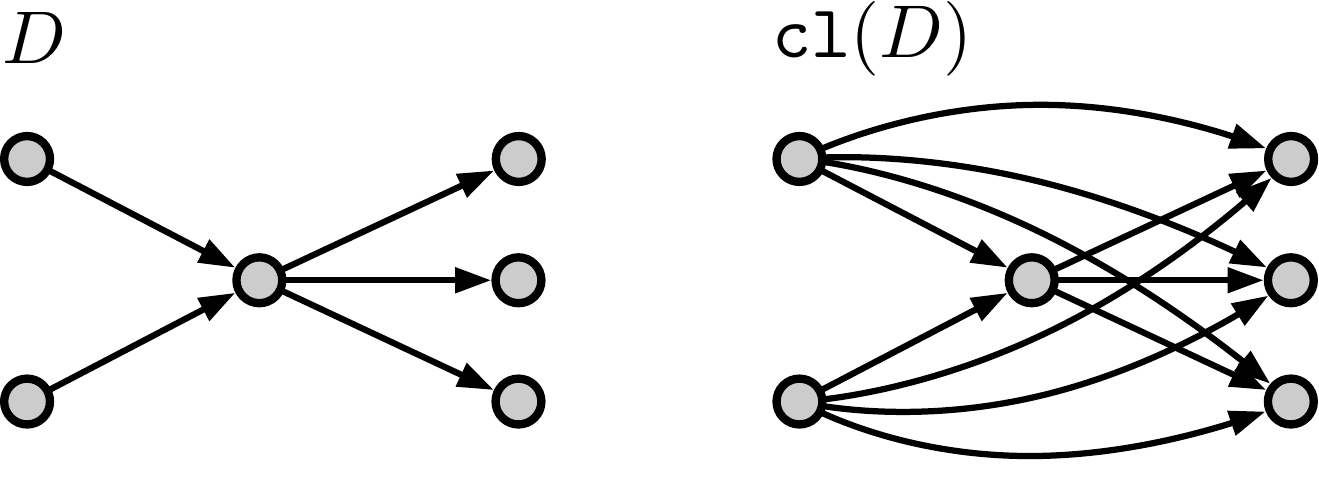}}
		\caption{Digraph $D$ and its transitive closure $\cl(D)$.}
		\label{fig:cl_example}
\end{center}\end{figure} 
\FloatBarrier
We say that a subgraph $H$ of $\mathtt{cl}(D)$ is \emph{feasible} if for each commodity $k \in \mathcal{K}$, there is an $s_k,t_k$-dipath in $H \cap \mathtt{cl}(P_k)$.  We define the min-degree sort point problem~as~follows. 
{\bf \mmspp}: given a directed graph $D$ and commodity set $\K$, find a feasible subgraph $H$ of $\cl(D)$ with minimum max out-degree. 

Given a graph $H$, let $\Delta^+(H)$ denote the maximum out-degree of a node in $H$. Let $\Delta^*$ denote the minimum max out-degree in any feasible subgraph. In the corresponding decision version of the problem, we are given a positive integer \emph{target}, $T$, and the problem is to determine whether or not $\Delta^* \leq T$. 

Let~$\mathcal{S} = \{s_1, s_2, \cdots, s_{k_s}\}$ be the set of sources, and $\T = \{t_1, t_2, \cdots, t_{k_t}\}$ be the set of sinks. For each $s \in \mathcal{S}$, let 
$\mathcal{T}(s) = \{t_k: (s, t_k) \in \K\}$ and for each~$t \in \mathcal{T}$, let $\mathcal{S}(t) = \{s_k: (s_k, t) \in \K\}$. 

We assume w.l.o.g.\  that $D$ is connected, the commodity set is non-empty {($\Delta^* \geq 1$),} all commodities are non-trivial ($s_k \neq t_k)$ and unique, and all nodes with no out-arcs in $D$ serve as the sink for some commodity. Any instance can be reduced to a (set of) equivalent instances that satisfy these assumptions. 


\subsection{Our results}
We say that an instance $\I = (D, \K)$ of \mmspp~is a \emph{tree instance} (\emph{star instance}) if the underlying undirected graph of $D$ is a tree (a star). We prove that \mmspp~is NP-hard, even when we restrict to the set of star instances.

\begin{restatable}{theorem}{hardness}\label{thm:hardness}
\mmspp~is NP-hard, even when restricted to star instances. 
\end{restatable}

Henceforth, we focus on tree instances of \mmspp. Note that in this class of  instances, any $s_k,t_k$-dipath in $\mathtt{cl}(D)$ can only use arcs in $\mathtt{cl}(P_k)$, where~$P_k$ is the unique $s_k,t_k$-dipath in~$D=(N,A)$. 

First, we construct combinatorial lower bounds that are motivated by the witness set construction for undirected min-degree spanning trees~\cite{FR_one}. Specifically, we define a function $\mathtt{LB}$ such that for all $W \subseteq N$ and $\K' \subseteq \K$, $\mathtt{LB}(W, \K') \leq \Delta^*$. We show that in instances with a single source, this construction is the best possible. We develop an exact polynomial-time local search algorithm for single-source instances and certify its optimality by determining values of $W$ and $\K'$ such that $\mathtt{LB}(W, \K') = \Delta^+(H)$ for the graph $H$ returned. Note that the single-source setting reduces to the problem of finding a min-degree arborescence in a directed acyclic graph. This problem can be solved via matroid intersection in~$O(n^3 \log n)$ time \cite{schrijver2003combinatorial,ChakrabartyLS0W19} or by a combinatorial algorithm in~$O(nm\log n)$ time~\cite{Yao}. Our approach uses the structure of the transitive closure, and allows us to obtain a very simple algorithm that beats the runtime of previous results by a quadratic factor. Moreover, it motivates our other algorithms in more complex settings that cannot be modeled as min-degree arborescence problems. 

\begin{restatable}{theorem}{singlesource}
There is an $O(n \log^2n)$-time exact algorithm for tree instances of \mmspp~with a single source.
\end{restatable}

When there is a single source and the undirected graph of $D$ is a tree, each node has at most one entering arc in $D$. We refer to such instances as \emph{out-tree} instances. We prove that $\max_{W \subseteq N, \K' \subseteq \K} \mathtt{LB}(W, \K') \geq \Delta^* - 1$ for out-tree instances, by showing that there is a polynomial-time algorithm that returns a solution with out-degree at most one greater than the best lower bound. We also show that there are out-tree instances where $\max_{W \subseteq N, \K' \subseteq \K} \mathtt{LB}(W, \K') = \Delta^* - 1$. 

Our polytime algorithm for multi-source out-tree instances is again combinatorial in nature, but the details are significantly more involved. In the algorithm for the single-source setting, in each iteration an arc~$vw$ is exchanged for an arc~$uw$, where $u$ is on the path between the root (the unique source) and $v$ in $D$. However, in the multi-source setting, the same action is not sufficient to ensure a high-quality solution is found. We instead define an algorithm for \mmspp~on out-trees which takes as input a target $T$, and returns a feasible solution~$H$ with~$\Delta^+(H) \leq T$ whenever $\Delta^* \leq T-1$. The additional input of the target as well as a careful selection of which arcs to select in each iteration prevent the need to backtrack. In the proof of the performance of this algorithm, we provide an explicit construction of the lower bound certificate with value at least $\Delta^* - 1$.

\begin{restatable}{theorem}{outtree}
There is a polynomial-time additive $1$-approximation algorithm for out-tree instances of \mmspp.
\end{restatable}

The analysis of the out-tree algorithm is \emph{tight}, in that there are instances of \mmspp~where the algorithm does not produce an optimal solution. Moreover, the performance of the algorithm is bounded against a lower bound that has an inherent gap matching the proven performance. The algorithmic approach for the out-tree setting cannot easily be extended to more complex graph structures, such as stars, since an optimal solution is no longer guaranteed to be acyclic and our current algorithm heavily relies on this fact. Furthermore, the lower bound construction weakens significantly for star instances. 

We also give a framework for obtaining  approximation results for arbitrary tree instances. As a first step, we give an efficient 2-approximation when $D$ is a star.
\begin{restatable}{theorem}{star}
There is a polynomial-time 2-approximation for star instances of \mmspp, and for any $\epsilon > 0$ there is a star instance of \mmspp~where $\max_{W \subseteq N, \K' \subseteq \K} \mathtt{LB}(W, \K') = \frac23 \Delta^* + \epsilon$. 
\end{restatable}
A generalization of the class of star instances is the class of junction trees. A tree instance of \mmspp~is a \emph{junction tree instance} if there is some node~$r$ in $D$ such that~$r$ is a node in $P_k$ for all $k \in \K$.
Our framework for obtaining an approximation result for general tree instances builds on the approximability of junction tree instances.
\begin{restatable}{theorem}{junction}
Given a polytime $\alpha$-approximation algorithm for junction tree instances, there is a polytime $\alpha \log n$-approximation algorithm for tree instances of \mmspp. 
\end{restatable}

%% file: related_work.tex
\subsection{Related work}

Degree-bounded network design problems are fundamental and well-studied combinatorial optimization problems. A prominent example is the \emph{minimum-degree spanning tree problem} which asks, given an undirected, unweighted graph, for a spanning tree with minimum maximum degree. Fürer and Raghavachari~\cite{FR_one} introduced the problem and presented a local-search based polynomial-time algorithm for computing a spanning tree with maximum degree which is at most $1$ larger than the  optimal maximum degree. Their combinatorial arguments rely on \emph{witness sets} chosen from a family of carefully-constructed lower bounds \cite{chvatal,win}. 

Various techniques have been employed in subsequent work on the weighted setting, e.g., \cite{Goemans06,jochen,jochen2,many_birds,ChaudhuriRRT09,ChaudhuriRRT06,RaviS06}, 
culminating in the result that one can compute a spanning tree of minimum cost that exceeds the degree bound by at most~$1$~\cite{Singh_Lau}. Since then, also generalizations have been studied such as the {\em degree-bounded Steiner tree} problem~\cite{KonemannR03,LauS13},  {\em survivable network design} with higher connectivity requirements~\cite{LauNSS09,LauS13} and the {\em group Steiner tree} problem~\cite{DehghaniEHL16,GuoKLLVX22,KortsarzN20}.

Directed (degree-bounded) network design problems are typically substantially harder than their undirected counterparts. Among the few nontrivial approximation results are quasipolynomial-time bicriteria approximations (with respect to cost and maximum out- or in-degree)~\cite{GuoKLLVX22} for the degree-bounded directed Steiner tree problem and approximation results for problems with intersecting or crossing supermodular connectivity requirements \cite{BKN09,Nutov}. 

A special case in directed degree-bounded network design is the \emph{min-degree arborescence problem} where, given a directed graph and root $r$, the goal is to find a spanning tree rooted at $r$ with minimum max out-degree. This problem is NP-hard \cite{BKN09, FR_one, LauNSS09} in general, and polytime solvable in directed acyclic graphs~\cite{Yao}.

%% file: hardness.tex
\section{Hardness}\label{sec:hardness}

In this section we prove that \mmspp~is NP-hard, even in the setting where the underlying undirected graph of $D$ forms a star. To prove this result, we will exhibit a reduction from the NP-hard problem of \hs~\cite{Garey}, defined as follows. 

\hs: Let $\Sigma = \{e_1, e_2, \ldots, e_m\}$ be a set of $m$ elements, let $\S = \{S_1, S_2, \ldots, S_n\}$ be a set of $n$ non-empty subsets of $\Sigma$, and let $b \in \mathbb{N}$ be a budget. The hitting set problem asks if there is a subset $R \subset \Sigma$ of cardinality at most $b$ such that $R \cap S_i \neq \emptyset$ for all $i \in [n]$. 

\hardness*

\begin{proof}
Let $\mathcal{H} = (\Sigma, \S, b)$ be an instance of \hs, where $b \in \mathbb{N}$, and $\Sigma = \{e_1, e_2, \ldots, e_m\}$, and~$\S = \{S_1, S_2, \ldots, S_n\}$ is a set of $n$ non-empty subsets of $\Sigma$. The hitting set problem asks if there is a subset $R \subset \Sigma$ of cardinality at most~$b$ such that $R \cap S_i \neq \emptyset$ for all $i \in [n]$. We will construct a corresponding instance~$\I = (D, \K)$ of \mmspp~such that $\mathcal{H}$ is a $\mathtt{YES}$ instance if and only if~$\Delta^* \leq c$, for some fixed integer $c$, where $\I$ and $c$ are polynomial in the size of~$\mathcal{H}$. 

First, we construct a digraph $D' = (N', A')$ with node and arc sets 
\begin{align*}
N' & = \{s_i: i \in [n]\} \cup \{v\} \cup \{t_j: j \in [m]\}, \\
A' & = \{s_i v: i \in [n]\} \cup \{v t_j: j \in [m]\}, 
\end{align*}
as shown in Figure \ref{subfig:hardnessA}. For each $i \in [n]$ and $e_j \in S_i$, we add a commodity with source $s_i$ and sink $t_j$ to form a set $\K'$. That is, $\K' = \{(s_i, t_j): i \in [n], e_j \in S_i\}.$

\begin{figure}[h!]
\vspace{-.2cm}
	\hspace{1.5cm}
  \begin{subfigure}{0.3\textwidth}
    \centering
    \includegraphics[width=\textwidth]{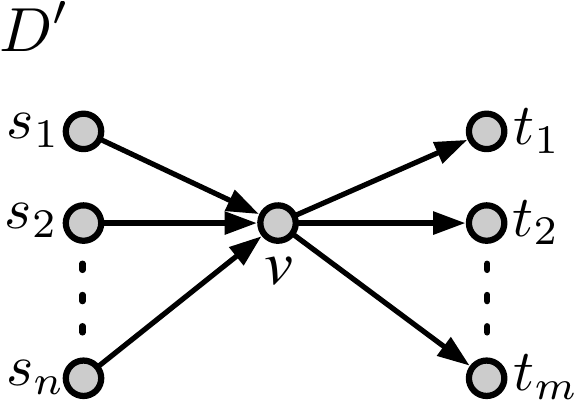}
    \caption{}
    \label{subfig:hardnessA}
  \end{subfigure}
  \hfill
  \begin{subfigure}{0.3\textwidth}
    \centering
    \includegraphics[width=\textwidth]{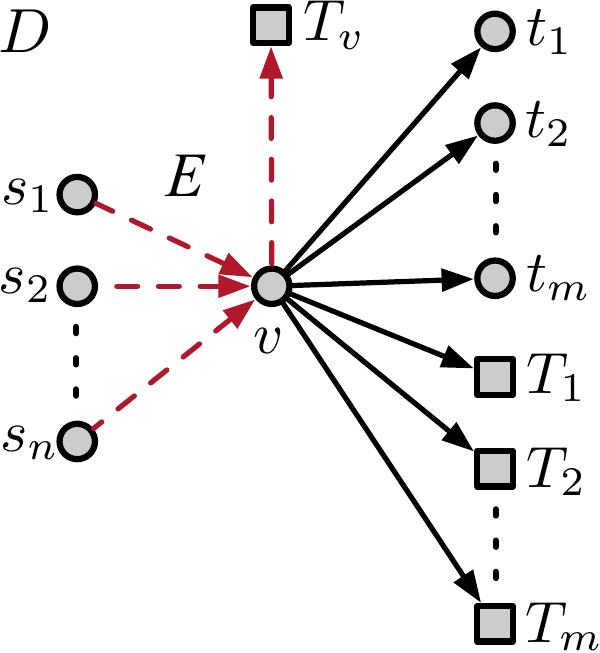}
	\vspace{-0.7cm}
    \caption{}
    \label{subfig:hardnessB}
  \end{subfigure}
\hspace{1.5cm}
\vspace{-.2cm}
  \caption{Digraphs $D'$ and $D$.}
  \label{fig:hardness}
\end{figure}
\FloatBarrier

We build on the instance $\I' = (D', \K')$ to form instance $\I = (D, \K)$. Let~$c = \max\{b, \max_i|S_i|\}$. For each $i \in [n]$, we add a set of nodes, denoted $T_i$, of cardinality $c - |S_i|$ to $N'$. Additionally, we add a set of nodes $T_v$ with cardinality~$c - b$. The arc set $A$ is formed by adding an arc from $v$ to all the nodes in $N\setminus N'$. That is,
\begin{align*}
N & = N' \cup T_v \cup \{T_i: i \in [n]\}\\
A & = A' \cup \{v t: t \in N' \setminus N\}.
\end{align*}
The digraph $D$ is given in Figure \ref{subfig:hardnessB}, where the square nodes indicate a set of nodes rather than a single node. Additionally, arcs entering a square node denote a set of arcs with same shared tail and differing heads -- one for each node in the set represented by the square node.

We add commodities to $\K'$ to form the set $\K$. For each $i \in [n]$, we add the commodity $(s_i, t)$ for all $t \in T_i$, and the commodity $(s_i, v)$. Additionally, for each node $t \in T_v$ we define a commodity $(v, t)$. Specifically, 
\begin{align*}
\K & = \K' \cup \{(s_i, v): i \in [n]\} \cup \{(s_i, t): t \in T_i, i \in [n]\} \cup \{(v,t): t \in T_v\}.
\end{align*}  

Observe that $\I$ and $c$ are polynomial in the input of $\mathcal{H}$. Furthermore, there are $c+1$ distinct commodities with source $s_i$ for each $i \in [n]$, and any feasible subgraph $H$ must contain the arc set $E = \{s_i v: i \in [n]\} \cup \{v t: t \in T_v\}$ (the red dashed arcs in Figure \ref{subfig:hardnessB}). We now claim that the hitting set instance $\mathcal{H}$ is a $\mathtt{YES}$ instance if and only if $\Delta^* \leq c$. 

First, suppose the hitting set instance $\mathcal{H}$ is a $\mathtt{YES}$ instance. Let $R$ be a set of at most $b$ elements such that $R \cap S_i \neq \emptyset$ for all $i \in [n]$. Up to reordering, we may assume that $R = \{e_1, e_2, \ldots, e_p\}$ where $p \leq b$. We form a feasible solution~$H$ as follows. First, let $H_1$ be the subgraph containing the edge set $E$ along with the set $\{v t_i: i \in [p]\}$, as shown in Figure \ref{fig:hard_sol}. 
We define the digraph $H_2$ as the subgraph of $\cl(D)$ with arc set $\{s_i t: t \in T_i, i \in [n]\} \cup \{s_i t_j: e_j \in S_i \setminus R, i \in [n]\}$. See Figure \ref{fig:hard_sol} for an example of $H_2$. 

\begin{figure}[htbp]
\vspace{-.2cm}
	\hspace{1.5cm}
  \begin{subfigure}{0.3\textwidth}
    \centering
    \includegraphics[width=\textwidth]{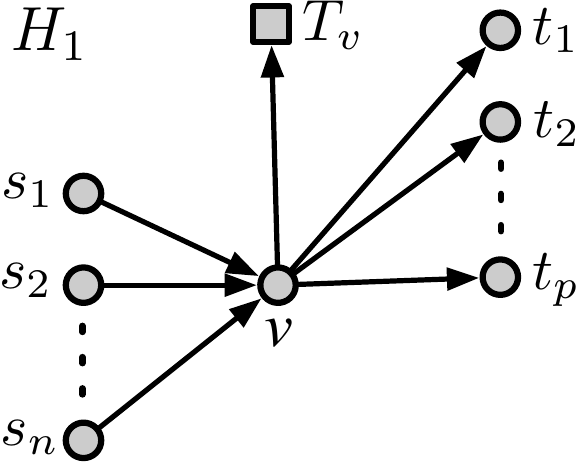}
    \label{subfig:hard_sol1}
  \end{subfigure}
  \hfill
  \begin{subfigure}{0.3\textwidth}
    \centering
    \includegraphics[width=\textwidth]{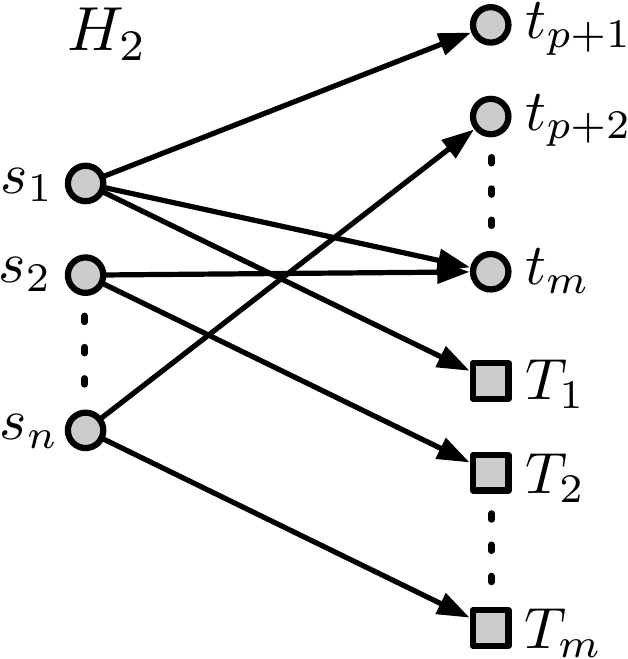}
	\vspace{-0.7cm}
    \label{subfig:hard_sol2}
  \end{subfigure}
\hspace{1.5cm}
\vspace{-.2cm}
  \caption{Digraphs $H_1$ and $H_2$.}
  \label{fig:hard_sol}
\end{figure}
\FloatBarrier

Let $H = H_1 \cup H_2$. We claim that $H$ is feasible and $\Delta^+(H) \leq c$. Let $k \in \K$. If $s_k  = v$, then $t_k \in T_v$ and $H_1$ contains an $s_k, t_k$-dipath. Otherwise, $s_k = s_i$ for some $i \in [n]$. If $t_k = t_j$ for some $e_j \in S_i \cap R$ or $t_k = v$, then similarly, $H_1$ contains an $s_k, t_k$-dipath. The final case is if $t_k = t_j$ for some $e_j \in S_i \setminus R$ or $t_k \in T_i$, in which case $H_2$ contains an $s_k, t_k$-dipath. Thus, $H$ is feasible. 

It remains to argue that $\Delta^+(H) \leq c$. For each $i \in [n]$, 
\begin{align*}
\deg^+_H(s_i) = 1 + \deg^+_{H_2}(s_i) = 1 + |T_i| + |S_i \setminus R| = 1+ c - |S_i| + |S_i \setminus R| \leq c, 
\end{align*}
where the inequality follows since $R \cap S_i \neq \emptyset$ for all $i \in [n]$. Additionally, we have \[\deg^+_H(v) = p + |T_v| \leq b + (c-b) = c.\] 
We now prove the reverse direction. Suppose $H \subseteq \cl(D)$ is a feasible subgraph for instance $\I$, with $\Delta^+(H) \leq c$. We may assume that $H$ is minimal in the sense that removing any arc results in an infeasible solution. This minimality ensures that for each $i \in [n]$, the heads of arcs with tail $s_i$ in $H$ must be in the set~${v \cup T_i \cup \{t_j: e_j \in S_i\}}$, which has cardinality $c+1$. 

Let $B$ be the set of arcs in $H$ from $v$ to some node in $\{T_i: i \in [n]\}$. We will argue that we may assume this set is empty. Suppose $B \neq \emptyset$. That is, $H$ contains an arc $v t$ for some $t \in T_i$ for some $i \in [n]$. By minimality,~${H - v t}$ is infeasible, and so $H$ does not contain the arc $s_i t$. If $\deg^+_H(s_i) < c$, then~${H - v t + s_i t}$ is a minimal feasible solution with one fewer node in $B$. Otherwise, $\deg^+_H(s_i) = c$. This implies that $H$ contains an arc $s t_j$ for some $e_j \in S_i$, since the set $S_i$ is non-empty. Then $H' = H - v t - s_i t_j + s_i t + v t_j$ is a minimal feasible solution with one fewer arc in $B$. By repeating this argument, we see that we may assume~$B = \emptyset$.

Since $H$ must contain the arc set $E$, and $\deg_E^+(v) = |T_v| = c - b$, $H$ contains at most $b$ arcs from the set $\{v t_j: j \in [m]\}$. Furthermore, for each $i \in [n]$ there are $c+1$ commodities (each with a distinct sink) which have $s_i$ as the source. Since $\deg_H(s_i) \leq c$, at least one of these commodities must be routed by a path of length two. That is, for each $s_i$ there is some commodity with source $s_i$ and sink $t$ where $v t$ is an arc in $H$. Since $B = \emptyset$, $H$ must contain an arc $v t_j$ for some $e_j \in S_i$. Therefore, $R = \{e_j: v t_j \in H, j \in [m]\}$ forms a hitting set of $\S$ of cardinality at most $b$. 
\end{proof}

%% file: lower_bounds.tex

\section{Combinatorial lower bounds}\label{sec:lower_bounds}

In this section we present a family of lower bounds, defined by a set of nodes $W$ and set of commodities $\K'$, which we refer to as a \emph{witness set}. 

Let $D = (N, A)$ be a digraph, and consider a subset $U \subseteq N$. We define $\delta^+_D(U)$ as the set of arcs in $D$ leaving $U$. That is, $\delta^+_D(U) := \{vw \in E: v \in U, w \notin U\}$. 

Consider a subset $W \subseteq N$ such that $s_k \in W$ and $t_k \notin W$ for some $k \in \mathcal{K}$. It follows that any feasible subgraph $H$ must contain some arc in $\delta^+_{\mathtt{cl}(D)}(W)$. Specifically, due to the given-path structure, we know that $H$ must contain some arc in $\delta^+_{\mathtt{cl}(P_k)}(W)$. In Figure \ref{fig:fan_ex3}, consider the node set $W = \{s, v\}$ along with the commodity with source $s$ and sink $t_1$. On the left is the base graph $D$, and on the right is the transitive closure. Since $s \in W$ and $t_1 \notin W$, it follows that any feasible solution $H$ must have a non-empty intersection with the set $\{s t_1, v t_1\}$. 

\begin{figure}[ht]
	\begin{center}
		\scalebox{0.36}{		
  \includegraphics{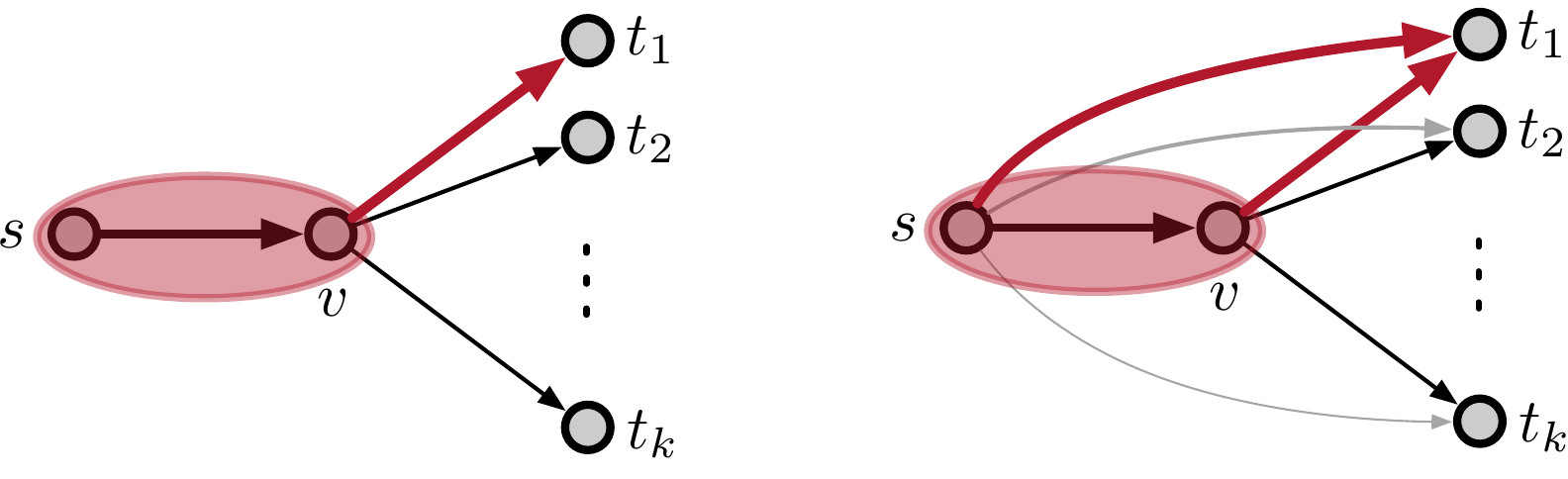}}
		\caption{A tree instance with $k$ commodities $\{(s, t_i): i \in [k]\}$.}
		\label{fig:fan_ex3}
\end{center}
\end{figure}
\FloatBarrier

We combine disjoint cuts for a fixed node set $W$ to obtain lower bounds on the value of $\Delta^*$. To simplify the set of cuts considered, we prove the following lemma to work with cuts in $D$ rather than in its closure. 

\begin{restatable}{lemma}{cuts}\label{lemma:cuts}
Let $\I = (D, \K)$ be a tree instance and $W\subseteq N$. For any $k, j \in \K$, 
\[\delta^+_{P_k}(W) \cap \delta^+_{P_j}(W) = \emptyset \mbox{ if and only if } \delta^+_{\mathtt{cl}(P_k)}(W) \cap \delta^+_{\mathtt{cl}(P_j)}(W) = \emptyset.\] 
\end{restatable}

\begin{proof} 
Suppose $\delta^+_{P_k}(W) \cap \delta^+_{P_j}(W) = \emptyset$ and for a contradiction, suppose there is an arc $vw \in \delta^+_{\cl(P_k)}(W) \cap \delta^+_{\cl(P_j)}(W)$. Then $v \in W$, $w \notin W$, and both~$P_k$ and~$P_j$ contain a $v,w$-dipath. Since the underlying undirected graph of $D$ is a tree, it follows that~$P_k$ and~$P_j$ contain the \emph{unique} $v,w$-dipath in $D$. Furthermore, since~$v \in W$ and~$w \notin W$, there is some arc $x y$ on the $v,w$-dipath where $x \in W$ and~$y \notin W$. As a result, $x y \in \delta^+_{P_k}(W) \cap \delta^+_{P_j}(W)$, a contradiction. 

The reverse direction immediately follows from the observation that~${\delta^+_{F} \subseteq \delta^+_{\cl(F)}}$ for any directed graph $F$. Thus, $\delta^+_{P_k}(W) \cap \delta^+_{P_j}(W) \subseteq \delta^+_{\mathtt{cl}(P_k)}(W) \cap \delta^+_{\mathtt{cl}(P_j)}(W)$. 
\end{proof}

As is standard in proving bounds on the min-max degree \cite{FR_one,win}, we observe the following:  if $\ell$ distinct arcs must leave a set $W$ of nodes in any feasible solution, then $\Delta^* \geq \lceil \ell / |W|\rceil$. The following lower bound construction shows that we can argue that such a disjoint arc set can be derived by looking at disjoint cuts of the form $\delta^+_{P_k}(W)$ for some $k \in \K$. Further, we show that this lower bound can be strengthened since we must also have connectivity \emph{within}~$W$ in order to allocate the arcs departing $W$ to different nodes in $W$. 

\begin{restatable}{lemma}{LBmulti}\label{lemma:LB_multi}
Let $\I = (D, \K)$ be a tree instance of \mmspp. Let $W \subseteq N$ such that $D[W]$ is connected and suppose $\emptyset \neq \mathcal{K}' \subseteq \mathcal{K}$ such that $s_k \in W$ and $t_k \notin W$ for all $k \in \mathcal{K}'$, and $\delta^+_{P_k}(W) \cap \delta^+_{P_j}(W) = \emptyset$ for all distinct $k, j \in \mathcal{K}'$. Then 
\[\Delta^* \geq \left\lceil \frac{|\mathcal{K}'| + |W| - |\S(\K')|}{|W|} \right\rceil,\]
where $\S(\K')$ denotes the set of sources for commodities in $\K'$.
\end{restatable}
\begin{proof}
Let $F$ be a feasible subgraph of $\cl(D)$, and let $W, \K'$ be given as in the statement. Let $H$ be the subgraph of $F$ where arcs are removed from $F$ if the remaining subgraph remains feasible for the commodity set $\K'$. That is, $H$ is a minimal subgraph of $F$ that contains an $s_k,t_k$-dipath for each commodity~$k \in \K'$. 

Since $s_k \in W$ and $t_k \notin W$, for each $k \in \K'$ it follows~${\delta^+_{\cl(P_k)}(W) \cap H \neq \emptyset}$. Since $\delta^+_{P_k}(W) \cap \delta^+_{P_j}(W) = \emptyset$ for all distinct $k, j \in \mathcal{K}'$, Lemma \ref{lemma:cuts} implies that for all distinct $k, j \in \mathcal{K}'$,~$\delta^+_{\mathtt{cl}(P_{k})}(W) \cap \delta^+_{\mathtt{cl}(P_{j})}(W) = \emptyset$. Thus, any feasible subgraph must have at least $|\mathcal{K}'|$ arcs in $\delta^+_{\mathtt{cl}(D)}(W)$, and so 
\[\Delta^* \geq \left\lceil \frac{|\mathcal{K}'|}{|W|} \right\rceil.\]
If any of the $|\mathcal{K}'|$ arcs in $\delta^+_{\mathtt{cl}(D)}(W) \cap H$ depart a node $v \in W \setminus \S(\K')$, then~$H$ contains an $s, v$-dipath for some source $s \in \S(\K')$, as otherwise $H$ was not minimal. Since $D[W]$ is connected and $\I$ is a tree instance, this $s,v$-dipath only uses arcs between nodes in $W$. Thus, if there are $\ell \geq 1$ nodes in $W$ with departing arcs in $\delta^+_{\mathtt{cl}(D)}(W) \cap H$, $H$ contains at least $\ell - |\S(\K')|$ arcs with both endpoints in $W$. Therefore, 
\[ \Delta^* \geq \min_{\substack{|\S(\K')| \leq \ell \leq |W|}} \ceil*{ \frac{|\mathcal{K}'| + \ell - |\S(\K')| }{\ell}} = \ceil*{ \frac{|\mathcal{K}'| + |W| - |\S(\K')| }{|W|}}.\]
For completeness, the above equality holds by the following argument. First we rewrite the inequality by letting $a = \ell - |\S(\K')|$.
\[ \min_{\substack{|\S(\K')| \leq \ell \leq |W|}} \ceil*{ \frac{|\mathcal{K}'| + \ell - |\S(\K')| }{\ell}} 
= \min_{\substack{0 \leq a \leq |W| - |\S(\K')|}} \ceil*{ \frac{|\mathcal{K}'| + a}{a + |\S(\K')|}}.\] 
Let $f(x) = \frac{|\mathcal{K}'| + x}{x + |\S(\K')|}$. Observe that for any $x \geq 0$, $f(x+1) \leq f(x)$, since 
\begin{align*}
~& |\S(\K')| \leq |\K'| \\
\Rightarrow ~& (|\K'| + x + 1) (x + |\S(\K')|) \leq (|\K'| + x) (x + 1 + |\S(\K')|) \\
\Rightarrow ~& \frac{|\K'| + x + 1}{x + 1 + |\S(\K')|} \leq \frac{|\K'| + x}{x + |\S(\K')|}.
\end{align*}
Thus, the expression is minimized when $\ell = |W|$.
\end{proof}
We define the functions $\mathtt{LB}$ and $\mathtt{LB}^w$, where for each $W \subseteq N$ and $\K' \subseteq \K$,  
\begin{align*}
\mathtt{LB}_\I(W, \K') & := \begin{cases}
\left\lceil \frac{|\K'| + |W| - |\S(\K')|}{|W|} \right\rceil & \mbox{if $W,\K'$ satisfy conditions of Lemma \ref{lemma:LB_multi} for $\I$}  \\
0 & \mbox{otherwise}
\end{cases} \\
\mathtt{LB}_\I^w(W, \K') & := \begin{cases}
\left\lceil \frac{|\K'|}{|W|} \right\rceil & \mbox{if $W,\K'$ satisfy conditions of Lemma \ref{lemma:LB_multi} for $\I$}  \\
0 & \mbox{otherwise}
\end{cases}
\end{align*}
We will drop the subscript $\I$ when the instance is clear from context. Note that $\mathtt{LB}^w$ is a weaker bound, in that it does not include any arcs due to connectivity within $W$. Since it holds~${|W| - |\S(\K')| \geq 0}$ for all inputs $W$ and $\K'$, we have the following corollary. 

\begin{corollary}
Suppose $\I = (D, \K)$ is a tree instance. For all $W \subseteq N$ and $\K' \subseteq \K$, \[\Delta^* \geq \mathtt{LB}^w(W, \K').\]
\end{corollary}

\begin{wrapfigure}{r}{.38 \textwidth}
\centering
		\scalebox{0.4}{
			\includegraphics{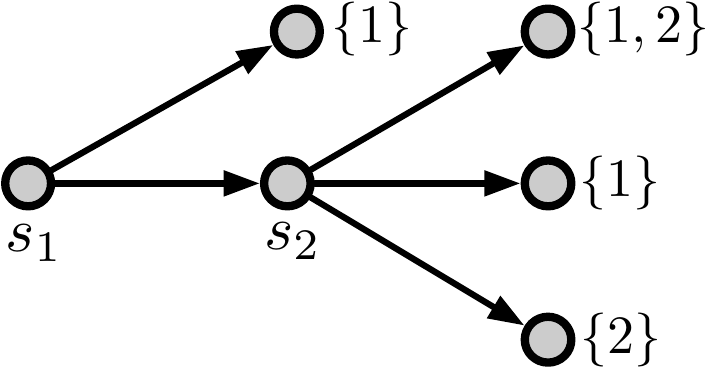}}
		\label{fig:gap}
\end{wrapfigure}

In the single-source setting, we will prove that for each instance $\I$, $\Delta^*$ is equal to $\max_{W \subseteq N, \K' \subseteq \K} \mathtt{LB}(W, \K')$.  A natural next step is to ask if the lower bound construction is also exact for multi-source out-tree instances. However, this is not the case, since there are instances of \mmspp~on out-trees where $\max_{W \subseteq N, \K' \subseteq \K} \mathtt{LB}(W, \K') = \Delta^* - 1$. Consider the out-tree instance on the right, where there are two sources $s_1$ and $s_2$. The remaining nodes are sinks labelled with each of the corresponding indices of sources that are matched to it. That is, a node $v$ labelled with $\{1,2\}$ indicates that there are commodities $(s_1, v)$ and $(s_2, v)$. In this instance, $\Delta^* = 3$, whereas $\max_{W \subseteq N, \K' \subseteq \K} \mathtt{LB}(W, \K') = 2$. 

While the witness set construction is not exact, it remains strong for out-tree instances. We prove in Section \ref{sec:out_trees} that for any (multi-source) out-tree instance, even the weaker bound has a gap of at most one:
\[ \max_{W \subseteq N, \K' \subseteq \K} \mathtt{LB}^w(W, \K') \geq \Delta^* -1.  \]

%% file: out_trees.tex
\section{(Near-)optimal algorithms for out-trees}\label{sec:out_trees}

We first present a simple combinatorial algorithm that solves single-source tree instances. (With a single source, $D$ necessarily forms an out-tree with root $s$.) We then present a polynomial-time algorithm for out-trees with multiple 
sources that returns a feasible solution with max out-degree at most one more than optimal. The algorithms are purely combinatorial and the analysis relies on bounds via the witness sets introduced in Section \ref{sec:lower_bounds}. 

\subsection{Simple algorithm for single-source setting}\label{subsec:single_out_trees}

We show that for each single-source instance $\I$, $\Delta^* = \max_{W \subseteq N, \K' \subseteq \K} \mathtt{LB}(W, \K')$. We prove this result by presenting a simple and efficient local search algorithm that returns an optimal solution, $H$, as well as a witness set $W, \K'$ such that $\Delta^+(H) = \mathtt{LB}(W, \K')$. \mnew{We then present a more efficient implementation of Algorithm \ref{alg:single_source_sketch} that runs in $O(n \log^2 n)$ time, beating previous results by a quadratic factor.} We state the results here and defer all details on the analysis to Appendix~\ref{app:single_source}.

\begin{algorithm}[H]
\DontPrintSemicolon
Start with the feasible solution $D$.\;
Select a max out-degree node $v^*$, and move an arbitrary arc departing~$v^*$ to the nearest predecessor $u$ of $v^*$ (in $D$) with~$\deg^+(u) \leq \deg^+(v^*)$.\;
Repeat Step 2 until no such predecessor exists.
\caption{Local search algorithm for the single-source setting.}
\label{alg:single_source_sketch}
\end{algorithm}

Figure \ref{fig:singlesource_ex} demonstrates the steps of the algorithm when the base graph $D$ is as provided in Figure \ref{subfig:a}, and $\T$ is the set of leaves. The square node is the selected max out-degree node, $v^*$, in each iteration and the nearest predecessor with degree at least two less than $v^*$ is $u$. Observe that in iteration 3, no such predecessor exists and so the algorithm terminates. 

\begin{figure}[h!]
  \begin{subfigure}{0.3\textwidth}
    \centering
    \includegraphics[width=\textwidth]{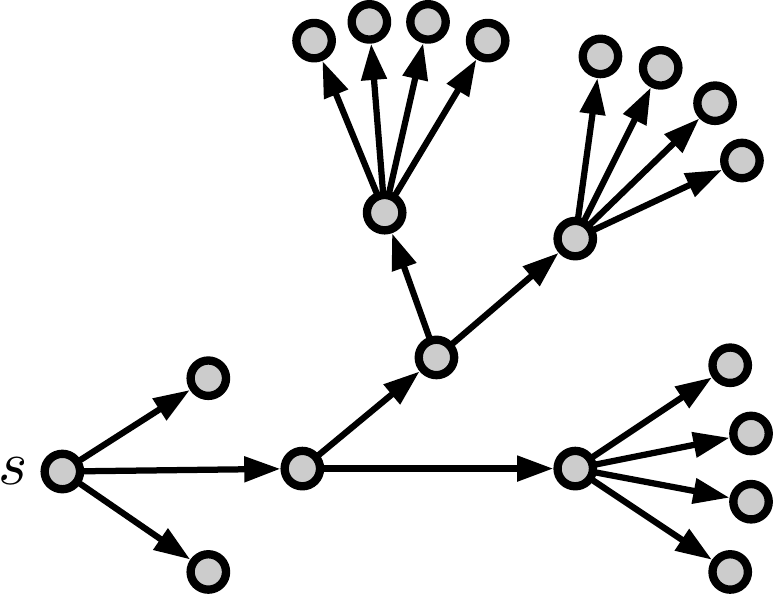}
    \caption{$H = D$}
    \label{subfig:a}
  \end{subfigure}
  \hfill
  \begin{subfigure}{0.3\textwidth}
    \centering
    \includegraphics[width=\textwidth]{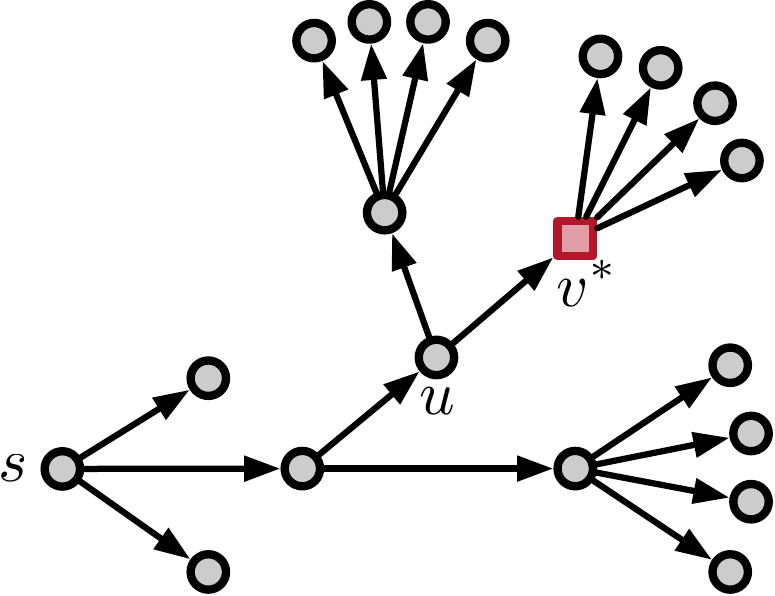}
    \caption{Iteration 1: $v^*$ identified.}
    \label{subfig:b}
  \end{subfigure}
  \hfill
  \begin{subfigure}{0.3\textwidth}
    \centering
    \includegraphics[width=\textwidth]{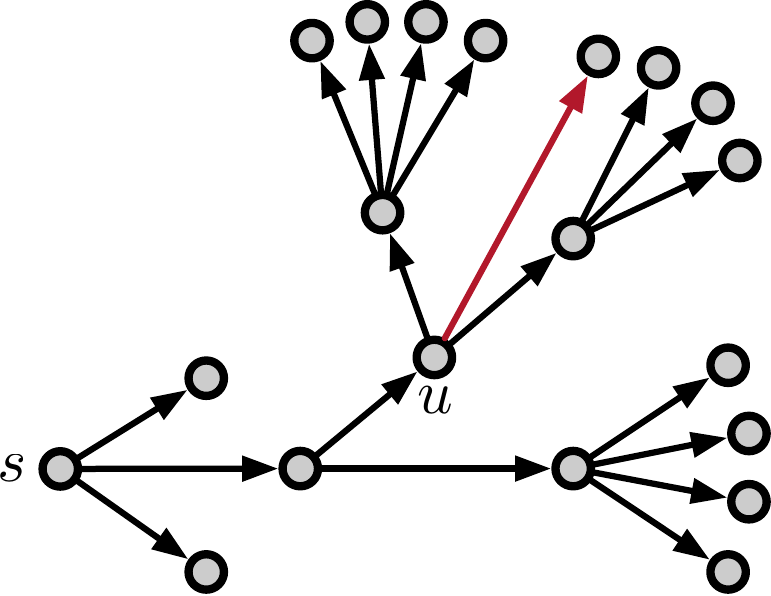}
    \caption{Iteration 1: $v^*w$ shifted.}
    \label{subfig:c}
  \end{subfigure}

  \medskip

  \begin{subfigure}{0.3\textwidth}
    \centering
    \includegraphics[width=\textwidth]{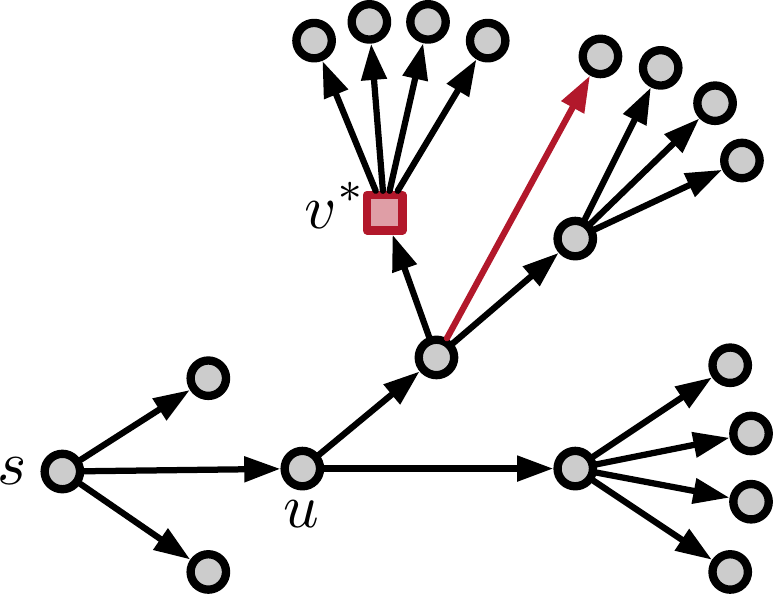}
    \caption{Iteration 2: $v^*$ identified.}
    \label{subfig:d}
  \end{subfigure}
  \hfill
  \begin{subfigure}{0.3\textwidth}
    \centering
    \includegraphics[width=\textwidth]{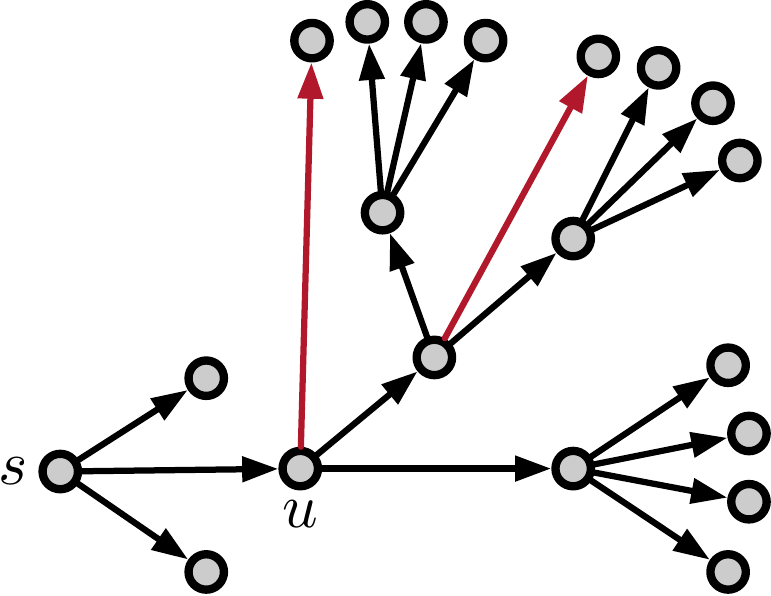}
    \caption{Iteration 2: $v^*w$ shifted.}
    \label{subfig:e}
  \end{subfigure}
  \hfill
  \begin{subfigure}{0.3\textwidth}
    \centering
    \includegraphics[width=\textwidth]{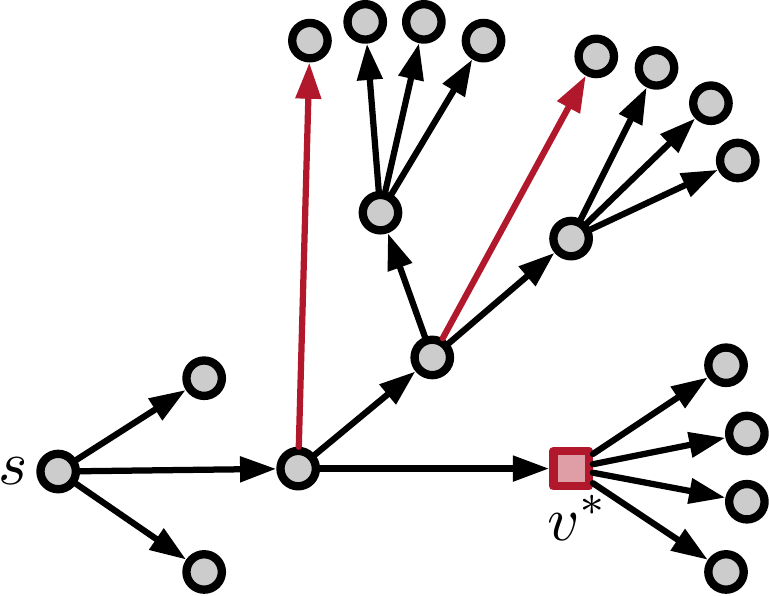}
    \caption{Iteration 3: $v^*$ identified.}
    \label{subfig:f}
  \end{subfigure}
  \caption{Execution of local search algorithm for single-source setting.}
  \label{fig:singlesource_ex}
\end{figure}
\FloatBarrier
We prove that this algorithm returns an optimal solution, $H$, by constructing a witness set $W, \K'$ such that $\mathtt{LB}(W, \K') = \Delta^+(H)$. 

\singlesource*

\subsection{Additive 1-approximation algorithm for out-trees}\label{subsec:multi_out_trees}

First, we describe how the multi-source setting differs from the single-source setting, making an extension of Algorithm \ref{alg:single_source_sketch} non-trivial. Consider the instance given in Figure \ref{subfig2:a}. The sinks are the set of leaves, and each sink $t$ is labelled with the set of indices of sources for which there is a commodity with that source and sink~$t$. Observe that we can no longer select departing arcs to shift arbitrarily: the arc $vw$ cannot be shifted since there is a commodity with source~$v$ and sink~$w$, while the other arcs departing $v$ can be shifted. Furthermore, it is no longer the case that in order to decrease the degree of the highest node, we only need to shift a single arc. In Figure \ref{subfig2:b} we see that in the potential second iteration of Algorithm \ref{alg:single_source_sketch}, no arc can be moved from $v$ since $s_2$ has out-degree 3, and all arcs departing $v$ serve commodities with source $s_2$.  

\begin{figure}[htbp]
  \begin{subfigure}{0.325\textwidth}
    \centering
    \includegraphics[width=\textwidth]{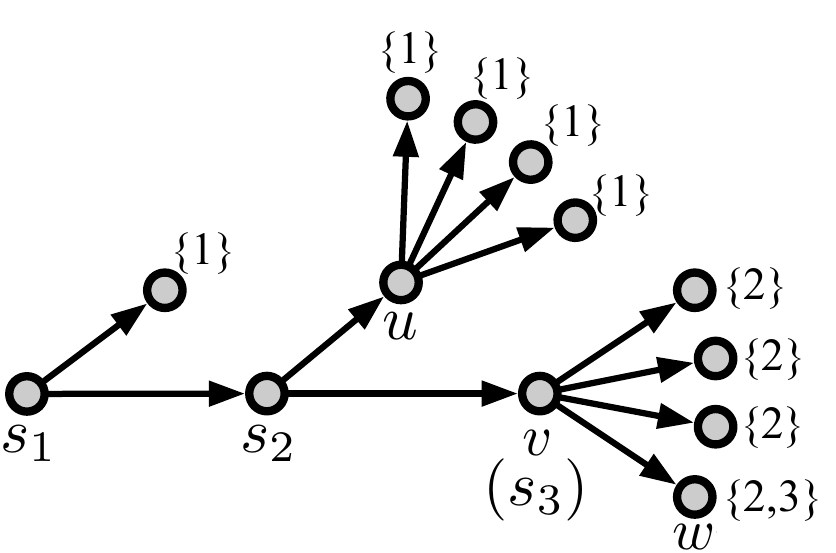}
    \caption{Out-tree instance}
    \label{subfig2:a}
  \end{subfigure}
  \hfill
  \begin{subfigure}{0.325\textwidth}
    \centering
    \includegraphics[width=\textwidth]{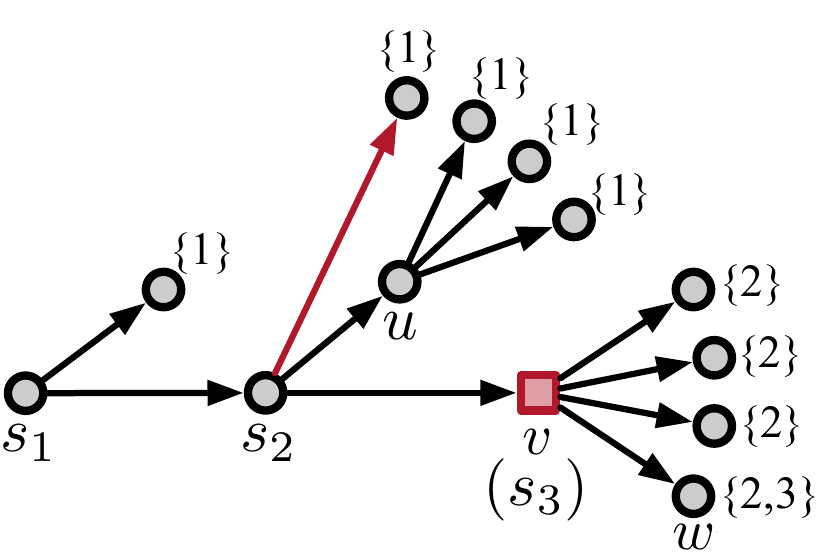}
    \caption{Terminal stage for alg.}
    \label{subfig2:b}
  \end{subfigure}
  \hfill
  \begin{subfigure}{0.325\textwidth}
    \centering
    \includegraphics[width=\textwidth]{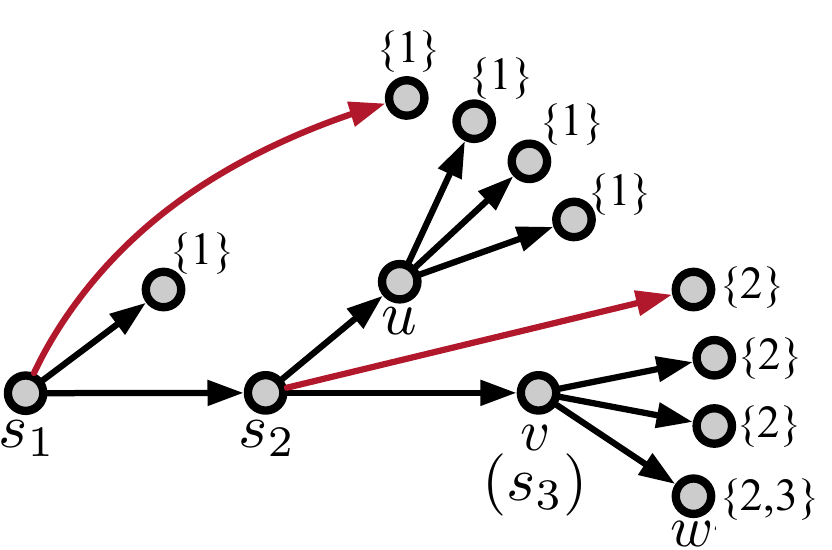}
    \caption{Solution ($\Delta^* = 3$).}
    \label{subfig2:c}
  \end{subfigure}
  \caption{Challenges in extending the single-source algorithm to multiple sources.} 
  \label{fig:extension_challenge}
\end{figure}
\FloatBarrier

\begin{definition}[$\S(a)$]
Let $D = (N,A)$ be an out-tree, and let $a \in A$. The \emph{set of sources that require $a$}, denoted $\S(a)$, is the set of sources $s_k$ for which $a$ is on the unique $s_k, t_k$-dipath in $D$. That is, $\S(a):= \{s_k: a \in A(P_k), k \in \K\}. $
\end{definition}

\begin{definition}[Blocking source]
Let $D = (N,A)$ be an out-tree and let \mbox{$a=vw$} be an arc in $A$. The \emph{blocking source} of $a$, denoted $b(a)$, is the unique source $s$ in $S(a)$ such that the $s, v$-dipath in $D$ has the fewest arcs.
\end{definition}

Consider the instance in Figure \ref{fig:blocking2}, where nodes $w_i$ for $i \in [5]$ are sinks, and each is labelled with the corresponding set of indices of sources. For each arc $a$ in $D$, the chart on the right indicates the set $\S(a)$ and value of $b(a)$.

\begin{figure}[h]
\vspace{-.5cm}
  \begin{minipage}{0.5\textwidth} 
    \centering
    \scalebox{0.35}{
	\includegraphics{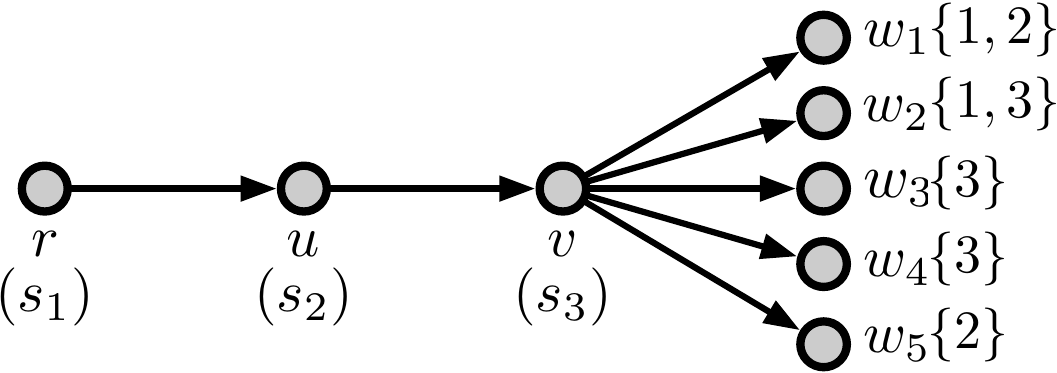}}
  \end{minipage}%
  \begin{minipage}{0.5\textwidth}
    \centering
\centering
\begin{tabular}{ccc}
$a$ & $\S(a)$ & $b(a)$ \\
\hline
$ru$ & $\{s_1\}$ & $s_1$ \\
$uv$ & $\{s_1, s_2\}$ & $s_2$ \\
$vw_1$ & $\{s_1, s_2\}$ & $s_2$ \\
$vw_2$ & $\{s_1, s_3\}$ & $s_3$ \\
$vw_3$ & $\{s_3\}$ & $s_3$ \\
$vw_4$ & $\{s_3\}$ & $s_3$ \\
$vw_5$ & $\{s_2\}$ & $s_2$ \\
\hline
\end{tabular}
  \end{minipage}
\caption{Example of blocking sources}
\label{fig:blocking2}
\end{figure}
\FloatBarrier

We will present an algorithm that takes as input a target, $T$, and returns a feasible solution $H$ with $\Delta^+(H) \leq T$ whenever $T \geq \Delta^* + 1$. We now define the contraction subroutine used to generate sub-instances of \mmspp.

\subsubsection{Contraction of an instance for target T}
Let $\I = (D, \K)$ be a feasible out-tree instance of \mmspp~with root $r$, and let $T \in \mathbb{Z}_{>0}$. For any node~${v \in N}$, let $N^+_D(v)$ be the set of nodes reached by arcs departing $v$ in $D$. 

Suppose there is more than one non-leaf node. Let $v$ be a non-leaf node where all nodes in $N^+_D(v)$ are leaves. That is, the subgraph of $D$ rooted at $v$ is a claw, denoted $C_v$. Such a node can be found efficiently by starting at $r$ and moving to a descendant that is not a leaf. Breaking ties arbitrarily, let $A^T_v$ denote the $\max\{T, |\delta^+_D(v)|\}$ arcs $a \in \delta^+_D(v)$ with the values of $b(a)$ with the fewest edges to $v$. Let $B^T_v = \delta^+_D(v) \setminus A^T_v$. We write $A_v$ and $B_v$ when $T$ is clear from context. In the example in Figure \ref{fig:blocking2}, we see that~$C_v$ is a claw, and when $T = 3$, $A_v = \{vw_2, vw_3, vw_4\}$ and $B_v = \{vw_1, vw_5\}$ since~$b(vw_2),b(vw_3), b(vw_4) = s_3 = v$ and $b(vw_1),b(vw_5)= s_2$ which is further from~$v$.

By definition, for any pair of arcs $a \in A_v$ and $a' \in B_v$, $b(a)$ is on the unique path in $D$ between $b(a')$ and $v$. Let $V(A_v)$ denote the heads of the arcs in $A_v$ and let $V(B_v)$ denote the heads of the arcs in $B_v$. For the same example and target, $V(A_v) = \{w_2, w_3, w_4\}$ and $V(B_v) = \{w_1, w_5\}$. 

We define the instance obtained from $\I$ by \emph{contracting $v$ for target $T$}, denoted~$\I^T_v = (D^T_v, \K^T_v)$ as follows. For all $v \neq r$, let~$p(v)$ denote the parent node of $v$ in~$D$.

\begin{definition}
The instance obtained from $\I$ by contracting $v$ for target $T$ is $\I^T_v = (D^T_v, \K^T_v)$, where
\begin{align*}
    E & := \{p(v)w: vw \in B_v\}, \\
    D^T_v & := \left(D \setminus \delta^+_D(v)\right) \cup E,
\end{align*}
and the commodity set is $\K^T_v = \{(s_k, t'_k): k \in \K\}$ where
\begin{align*}
	t'_k & = \begin{cases}
v, \quad ~\mbox{if } t_k \in V(A_v)\\
t_k, \quad \mbox{otherwise.}
\end{cases}
\end{align*}
\end{definition}
An example of this contraction process is given in Figure \ref{fig:target_T}. The instance is given on the left, where the sinks are only labelled for the nodes in the claw $C_v$, since this is the only set of commodities impacted by the contraction of $v$. $\I^3_v$ is given on the right.

\begin{figure}[h!]
	\begin{center}
		\scalebox{0.37}{
			\includegraphics{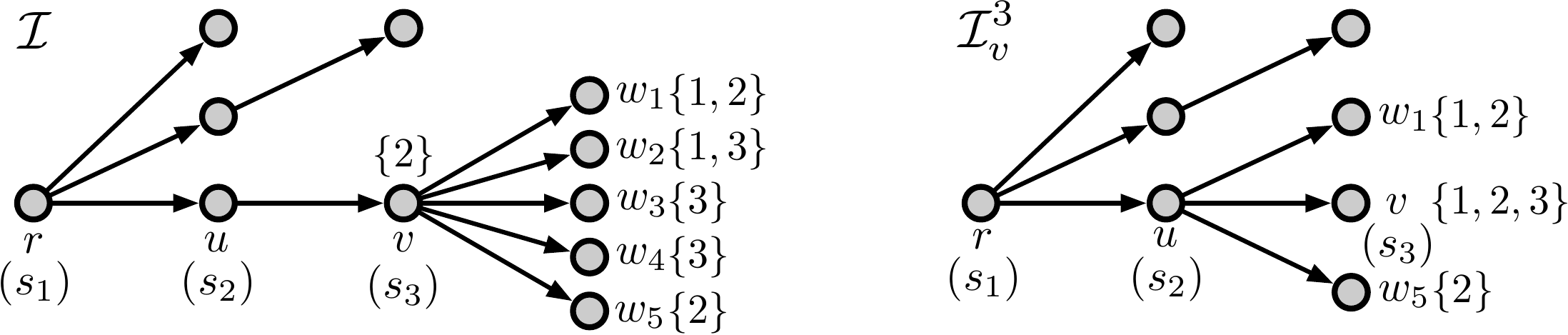}}
		\caption{Example for contraction process.}
		\label{fig:target_T}
\end{center}
\end{figure} 

Recall that for a node $v \in N$, $\T(v)$ denotes the set of sinks among commodities with source $v$. That is, $\T(v) = \{t_k: s_k = v, k \in \K\}$. 
\begin{restatable}{lemma}{welldefined}\label{lemma:well_defined}
If $T \geq |\T(v)|$, then $\I^T_v$ is a feasible instance of \mmspp.
\end{restatable}

\begin{proof}
We need to show that for each commodity $k \in \K^T_v$, there is an $s_k, t'_k$-dipath in $D^T_v$. If $t'_k \notin V(C_v)$, then $t'_k = t_k$, and $(s_k, t_k)$ was a commodity in $\I$. The $s_k, t_k$-dipath in $D$ did not use any nodes in $C_v$ and so it remains in $D^T_v$. 

If $t'_k = v$, then there is some commodity $(s_k, t_k)$ in $\I$ where $t_k \in V(C_v)$. Since $s_k \neq t_k$, the $s_k, t_k$-dipath in $D$ consists of an $s, v$-dipath as a subpath. This dipath is unchanged in $D_v^T$. 

Finally, suppose $t'_k \in V(C_v) \setminus v$. By construction, $t_k \in V(B_v)$. Then $t'_k = t_k$, and $(s_k, t_k)$ is a commodity in $\I$, and $D$ contains an~$s_k, t_k$-dipath consisting of an~$s_k, v$-dipath along with the arc $v t_k$. However, when forming $D_v^T$ we remove the arc $v t_k$ and replace it with the arc $u t_k$, and so the same dipath will not suffice. However, since $T \geq |\T(v)|$, it follows that $s_k \neq v$. Thus, the $s_k, v$-dipath is non-trivial and consists of an $s_k, u$-dipath along with the arc $uv$. Then, the~${s_k, u}$-dipath along with the arc $u t_k$ forms an~$s_k, t_k$-dipath in $D_v^T$ as required.  
\end{proof}


\subsubsection{Algorithm}
For a given a target $T$, 
our algorithm returns a feasible solution $H$ with $\Delta^+(H) \leq T$ whenever $T \geq \Delta^* + 1$. To guarantee the performance of the algorithm if no such solution is produced (the algorithm outputs the empty set), we show that in this case there is a witness set $W, \K'$ such that $\LB^w(W, \K') \geq T$. 

When there is a single non-leaf node, $r$, either the graph itself is the desired solution with max out-degree at most $T$, or the set $W = \{r\}, \K' = \K$ is a witness set certifying that $\Delta^* \geq T$. Otherwise, we find a node, $v$, with only leaves as descendants. Then either the subtree rooted at $v$ already provides a witness set, or we apply the target algorithm to the instance $\I_v^T$. If the algorithm produces a feasible solution $H_v$ to $\I_v^T$ with $\Delta^+(H_v) \leq T$, we then extend this subgraph to a feasible solution $H$ for $\I$ with $\Delta^+(H) \leq T$ by simply adding back the arcs in the set $A_v$. The pseudocode is presented in Algorithm \ref{alg:out-trees}.

\begin{algorithm}[h!]
\setcounter{AlgoLine}{0}
	\DontPrintSemicolon
	\KwIn{A target $T \in \mathbb{Z}_{\geq 1}$, and an out-tree instance of \mmspp, $\I = (D, \K)$}
	$L \leftarrow \{v \in N: \delta^+_D(v) = \emptyset\}$\\
	$I \leftarrow N \setminus L$\\
	\If{$|I| = 1$}{
		\If{$|L| > T$}{
			\Return $\emptyset$}
		\Return $D$
	}
	Let $v$ be a non-leaf where $N^+_D(v) \subseteq L$\\
	\If{$|\T(v)| > T$}{
		\Return $\emptyset$}
	$H_v \leftarrow \mathtt{out}$-$\mathtt{tree}(\I^T_v, T)$ \\
	\If{$H_v = \emptyset$}{
		\Return $\emptyset$}
	\Else{
		\Return $H_v \cup A_v$}
	\caption{$\mathtt{out}$-$\mathtt{tree}(\I, T)$}\label{alg:out-trees}
\end{algorithm}
\FloatBarrier

If Algorithm \ref{alg:out-trees} returns the empty set, we argue by induction that there is a witness set $W_v, \K'_v$ for the instance $\I^T_v$ such that $\LB_{\I_v^T}^w(W_v, \K'_v) \geq T$. We then extend this pair to a witness set $W, \K'$ for $\I$ such that $\LB_{\I}^w(W, \K') \geq T$. Note that we cannot necessarily set $W = W_v$ and $\K' = \K'_v$, since the commodity paths may differ in $\I$ and $\I_v^T$. For each commodity $k$, the corresponding source-sink pair is $(s_k, t_k)$ in $\I$ and $(s_k, t'_k)$ in $\I_v^T$. Let $Q_k$ denote the unique $s_k,t_k'$-dipath in $D_v^T$, and recall that $P_k$ denotes the unique $s_k, t_k$-dipath in $D$. The following lemmas relate the cut sets in $\I$ and $\I_v^T$ for a fixed commodity and node set. 
%
%
\begin{restatable}{lemma}{compaths}\label{lemma:com_paths}
Let $X \subseteq N \setminus \{v\}$ such that $D[X]$ is connected, and let $k \in \K$. If $s_k \in X$ and $t'_k \notin V(B_v)$, then $\delta^+_{P_k}(X) = \delta^+_{Q_k}(X)$.
\end{restatable}

\begin{proof}
If $t_k = v$, then $t_k' = t_k$ and $P_k$ did not contain any arc from~$A_v$ or~$B_v$. Therefore, $P_k$ and $Q_k$ are the same, and so $\delta^+_{P_k}(X) = \delta^+_{Q_k}(X)$. The same argument holds when $t'_k \neq v$.

Otherwise, $t'_k = v$ and $t_k \neq v$. It follows that $t_k \in V(A_v)$ and so $Q_k$ is a subpath of $P_k$, giving $\delta^+_{Q_k}(X) \subseteq \delta^+_{P_k}(X)$. Since $v \notin X$ and $s_k \in X$, $\delta^+_{Q_k}(X) \neq \emptyset$. Moreover, $|\delta^+_{P_k}(X)| \leq 1$ in general, so $\delta^+_{P_k}(X) = \delta^+_{Q_k}(X)$. 
\end{proof}

\begin{restatable}{lemma}{compathstwo}\label{lemma:com_paths2}
Suppose $|\T(v)| \leq T$. Let $X \subseteq N$ such that $D[X]$ is connected, and let $k \in \K$.
If $s_k \in X$, $t'_k \in V(B_v)$, and $t'_k \notin X$, then either $\delta^+_{P_k}(X) = \delta^+_{Q_k}(X)$ or $\delta^+_{Q_k}(X) = p(v) t'_k$. 
\end{restatable}

\begin{proof}
Since $|\T(v)| \leq T$ and $t_k' \in V(B_v)$, $s_k$ cannot be equal to $v$. Therefore, each of $P_k$ and $Q_k$ contain the same $s_k, p(v)$-dipath in $D$, denoted $R$, as a subpath. An example is shown in Figure \ref{fig:com_paths2}, for the commodity with source $r$ and sink~$w_1$. The solid edges form $P_k$, and the dashed lines form $Q_k$. The subpath $R$ is the dipath that is shared by $P_k$ and $Q_k$ (the single arc $ru$ in this case).

\begin{figure}[htbp]
  \begin{subfigure}{0.43\textwidth}
    \centering
    \includegraphics[width=\textwidth]{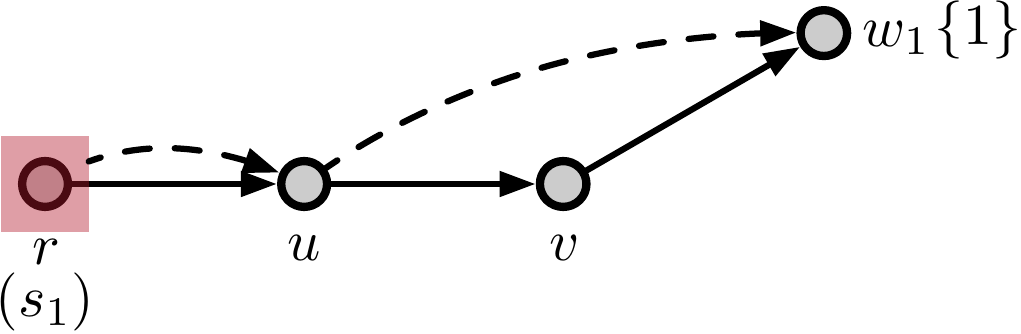}
    \caption{$\delta^+_{R}(X) \neq \emptyset$}
    \label{subfig:com_paths2A}
  \end{subfigure}
  \hfill
  \begin{subfigure}{0.43\textwidth}
    \centering
    \includegraphics[width=\textwidth]{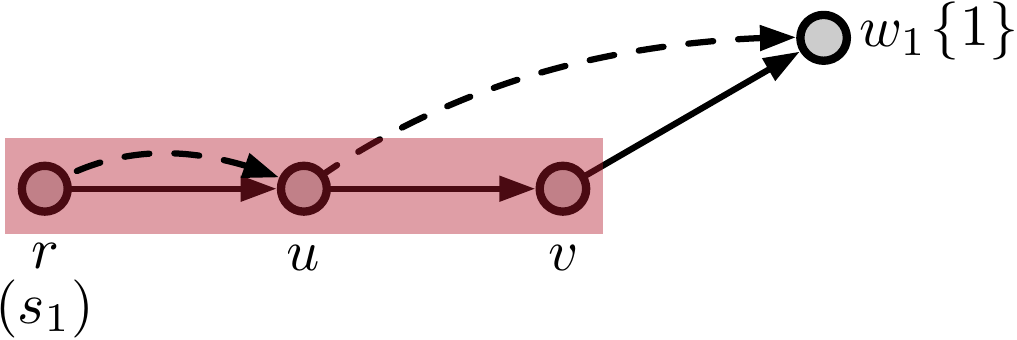}
    \caption{$\delta^+_{R}(X) = \emptyset$}
    \label{subfig:com_paths2B}
  \end{subfigure}
  \caption{Comparison of $\delta^+_{Q_k}(X)$ and $\delta^+_{P_k}(X)$.}
  \label{fig:com_paths2}
\end{figure}
\FloatBarrier

If $\delta^+_{R}(X) \neq \emptyset$, then $\delta^+_{Q_k}(X) = \delta^+_{R}(X) = \delta^+_{P_k}(X)$. So suppose $\delta^+_{R}(X) = \emptyset$. Since~$t'_k \notin X$, it follows that $\delta^+_{Q_k}(X) \neq \emptyset$. The only arc in $Q_k \setminus R$ is the arc~$p(v) t'_k$. 
\end{proof}
%
%
%

\begin{prop}
For each $T \in \mathbb{Z}_{>0}$, Algorithm \ref{alg:out-trees} either returns a solution, $H$, to instance $\I$ with $\Delta^+(H) \leq T$, or there is a witness set $W, \K'$ such that $\LB^w(W, \K') \geq T$, certifying that $\Delta^* \geq T$. 
\end{prop}

\begin{proof}
Let $\I = (D, \K)$ denote an instance of \mmspp~where $D$ is an out-tree with root $r$. Let $T$ be a positive integer. We prove the result by induction on the number of non-leaf nodes in $D$. Recall that we may assume that we are working with minimal, non-trivial instances, in the sense that all $k \in \K$ have~$s_k \neq t_k$, $\K \neq \emptyset$, and any leaf node $l$ must be a sink for some commodity. 

Suppose there is a single non-leaf node, $r$, as in Figure \ref{fig:claw}. Then each node in $\delta^+_D(r)$ is a leaf. If $|\delta^+_D(r)| \leq T$, then the graph $D$ is a solution with $\Delta^+(D) \leq T$. Otherwise,~${|\T(r)| = |N^+_D(r)| > T}$, and the witness set $W = \{r\}$, $\K' = \K$ gives the lower bound $\LB^w(W, \K') = \left\lceil |\T(r)|/ 1 \right\rceil \geq T$. 
\begin{figure}[h!]
	\begin{center}
		\scalebox{0.4}{
			\includegraphics{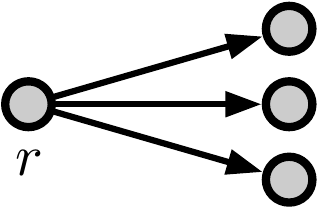}}
		\caption{Single non-leaf node.}
		\label{fig:claw}
\end{center}
\vspace{-.5cm}
\end{figure} 
\FloatBarrier
Suppose the result holds for all out-tree instances with at most $d$ non-leaf nodes, and consider an instance with $d+1$ non-leaf nodes. Let $v$ be a non-leaf node where all nodes in $N^+_D(v)$ are leaves. 

If $|\T(v)| > T$, then setting $W = \{v\}$ and $\K' = \{k \in \K: s_k = v\}$ gives $\LB^w(W, \K') > T$, so suppose~${|\T(v)| \leq T}$. Let $A_v$, $B_v$, $E$, and $\I_v^T = (D_v^T, \K_v^T)$ be defined as above, where~${A_v \cup B_v = \delta^+_D(v)}$ and for any pair $a \in A_v$ and $a' \in B_v$, $b(a)$ is on the unique path in $D$ between $b(a')$ and $v$. By Lemma \ref{lemma:well_defined}, it follows that $\I^T_v$ is a feasible out-tree instance with at most $d$ non-leaf nodes. Let $\Delta^*_v$ denote the min-max degree of a feasible subgraph of $\cl(D_v)$ for instance $\I^T_v$. By induction, the algorithm either returns a feasible subgraph $H_v \subseteq \cl(D_v)$ for $\I_v^T$ with $\Delta^+(H_v) \leq T$, or there is a witness set $W_v, \K'_v$ certifying that $\Delta^*_v \geq T$.  

Suppose the algorithm returns a feasible subgraph $H_v \subseteq \cl(D_v)$ for $\I^T_v$ with $\Delta^+(H_v) \leq T$. Let $H = H_v \cup A_v$, as in the example in Figure \ref{fig:extending_feasible}. We claim that~$H$ is a feasible solution for $\I$ with $\Delta^+(H) \leq T$. Observe that $\deg^+_{H_v}(v) = 0$ and $|A_v| \leq T$, so $\Delta^+(H) \leq T$. It remains to argue that $H$ is feasible. Let $k \in \K$. If $t_k \notin V(A_v)$, then $(s_k, t_k)$ is a commodity in $\I^T_v$, and so an $s_k, t_k$-dipath is present in $H_v \subseteq H$. Otherwise $t_k \in V(A_v)$, and so there is a commodity $(s_k, v)$ in $\K^T_v$. By feasibility of $H_v$, $H$ contains an $s_k, v$-dipath. With the addition of arc $v t_k$ in~$A_v$ we obtain an $s_k, t_k$-dipath in $H$. Therefore, $H$ is feasible and $\Delta^+(H) \leq T$. 

\begin{figure}[h!]
	\begin{center}
		\scalebox{0.38}{
	\includegraphics{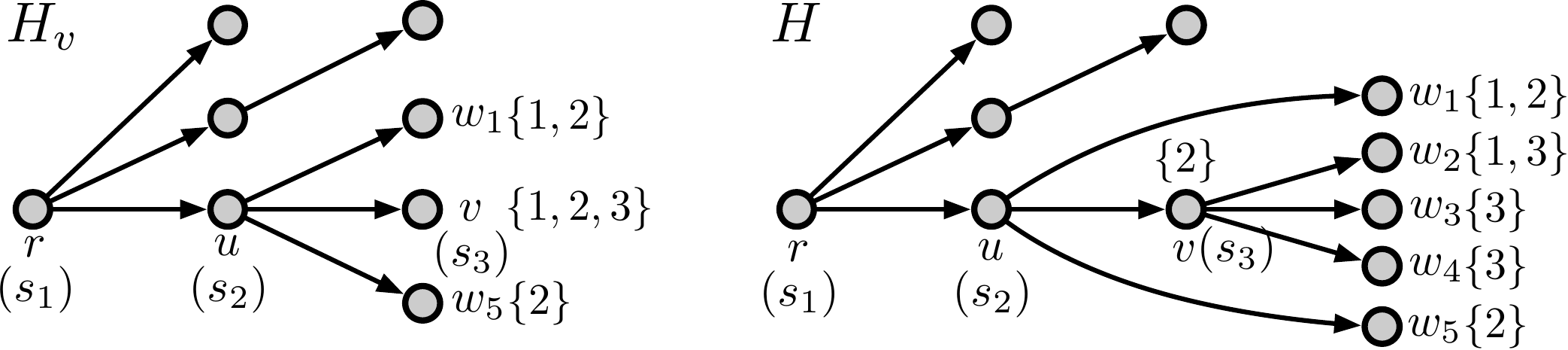}}
		\caption{Extending feasible $H_v$ for $\I_v^3$ to a feasible solution $H$ for $\I$.}
		\label{fig:extending_feasible}
\end{center}
\end{figure} 
\FloatBarrier


Suppose instead, there is a witness set $W_v, \K'_v$ with $\LB^w_{\I_v^T}(W_v, \K'_v) \geq T$. That is, $D_v^T[W_v]$ is connected, $\emptyset \neq \K'_v \subseteq \K_v^T$ such that $s_k \in W_v$ and $t'_k \notin W_v$ for all~$k \in \K'_v$. Moreover, for all distinct $k, j \in \mathcal{K}_v'$, $\delta^+_{Q_k}(W) \cap \delta^+_{Q_j}(W) = \emptyset$. Note that we may assume $W_v$ does not contain any leaves in $D_v^T$, since removing a leaf from~$W_v$ can only increase the lower bound. So, $W_v \subseteq N \setminus \{v\}$.

We now construct a witness set $W, \K'$ for $\I$ such that $\Delta^* \geq \LB^w_\I(W, \K') \geq T$. Consider the same set of nodes $W_v$ and commodities $\K_v'$ in $D$ instead of $D_v^T$. Since no leaf node in $D_v^T$ is in $W_v$, $D[W_v]$ is connected. Moreover, for each~$k \in \K_v$, it holds~$s_k \in W_v$ and $t_k \notin W_v$. 

If $\delta^+_{P_k}(W) = \delta^+_{Q_k}(W)$ for all $k \in \K'_v$, then it follows that $W = W_v$ and~$\K' = \K'_v$ gives $\LB_\I^w(W, \K') = \LB_{\I^T_v}^w(W, \K'_v) \geq T$ as required. So, suppose there is some commodity $k \in \K'_v$ such that $\delta^+_{P_k}(W) \neq \delta^+_{Q_k}(W)$. By Lemmas \ref{lemma:com_paths} and \ref{lemma:com_paths2}, it follows that there is a commodity $\ell \in \K'_v$ with $t'_{\ell} \in V(B_v)$ and $p(v) \in W_v$. 

Let $W = W_v \cup \{v\}$. Since $p(v) \in W_v$ and $D[W_v]$ is connected, it follows that $D[W]$ is connected. Let $\bar{\K} = \{(b(v t_k), t_k): t_k \in V(A_k)\}$. By definition of $\bar{\K}$ and $B_v$, it follows that $s_k$ is on the unique dipath between $s_{\ell}$ and $v$ for all $k \in \bar{\K}$. Since $D[W]$ is connected and contains both $s_{\ell}$ and $v$, it follows that~$s_k \in W$ for all $k \in \bar{K}$. 

If there is a commodity $k \in \K'_v$ with $t_k = v$, let $k_v$ denote this commodity, and otherwise let $k_v$ denote the empty set. Let $\K' = (\K'_v \cup \bar{\K}) \setminus k_v$. Then $s_k \in W$ for all $k \in \K'$, and $t_k \notin W$ for each $k \in \K'$ since $k_v$ was removed. We also claim that $\delta^+_{P_k}(W) \cap \delta^+_{P_j}(W) = \emptyset$ for any $k, j \in \K'$. If $k \in \K'$ with $t_k \in N_D^+(v)$, then $\delta^+_{P_k}(W) = v t_k$. Otherwise, $t_k \notin V(C_v)$ and $\delta^+_{P_k}(W) = \delta^+_{Q_k}(W) \notin \delta^+_D(v)$. Therefore, the disjointness of the cuts follows.

In order to prove that $\LB^w_\I(W, \K') \geq T$, it remains to argue that the cardinality of $|\K'|$ is at least $(T-1)|W| + 1$. Observe that by definition, if $B_v \neq \emptyset$, then $|A_v| = T$. By induction, we have~$|\K'_v| \geq (T-1)|W_v| + 1$. Therefore, $|\K'| \geq (T-1)|W_v| + T - 1 + 1 = (T-1)|W| + 1$, and so $\LB^w_\I(W, \K') \geq T$ as required. 
\end{proof}

\outtree*

\begin{proof}
By executing Algorithm \ref{alg:out-trees} for all possible target values, we determine the smallest value of $T$ for which Algorithm \ref{alg:out-trees} gives a feasible subgraph $H$ with~$\Delta^+(H) \leq T$. Since Algorithm \ref{alg:out-trees} did not return a feasible subgraph for the target~$T - 1$, there is a witness set $W, \K'$ with~$\LB^w(W, \K') \geq T-1$ (constructible in polytime), certifying that $\Delta^* \geq T-1$. Hence, $\Delta^+(H) \leq \Delta^* + 1$. 
\end{proof}

As previously mentioned in the discussion of lower bounds, there are out-tree instances of \mmspp~where $\max_{W \subseteq N, \K' \subseteq \K} \mathtt{LB}(W, \K') = \Delta^* - 1$. Furthermore, the analysis of Algorithm \ref{alg:out-trees} is tight since there are out-tree instances for which the algorithm does not return a solution with out-degree at most~$T$, even if $\Delta^* = T$.

Consider the following instance in Figure \ref{subfig:broom_gapA}. A solution $H$ with $\Delta^+(H) = 2$ is given in Figure \ref{subfig:broom_gapB}.  
\begin{figure}[htbp]
  \begin{subfigure}{0.45\textwidth}
    \centering
    \includegraphics[width=\textwidth]{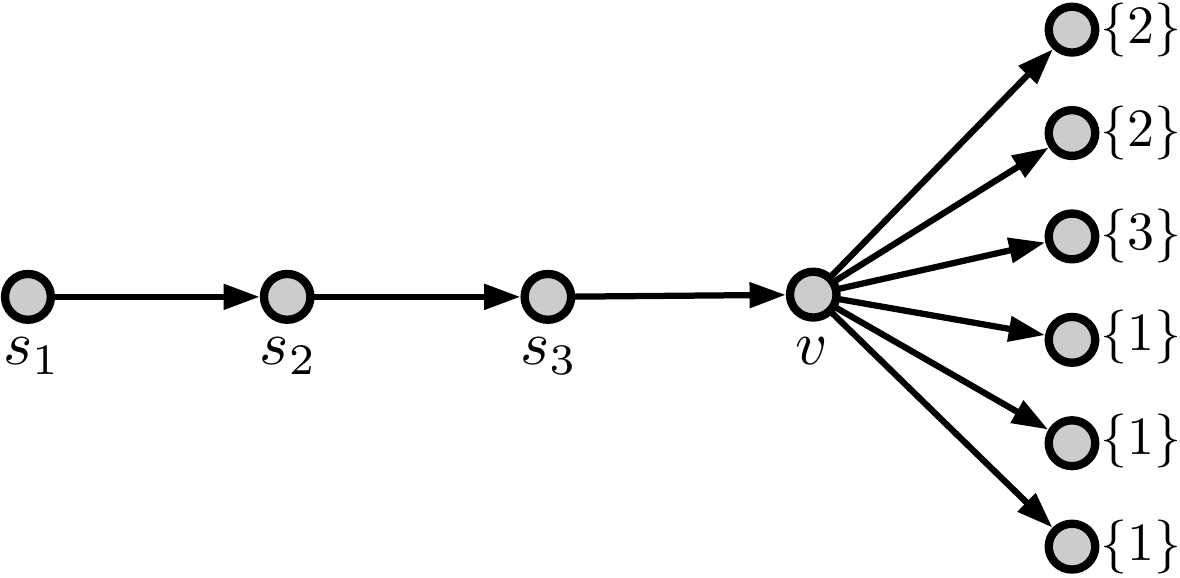}
    \caption{Out-tree instance.}
    \label{subfig:broom_gapA}
  \end{subfigure}
  \hfill
  \begin{subfigure}{0.45\textwidth}
    \centering
    \includegraphics[width=\textwidth]{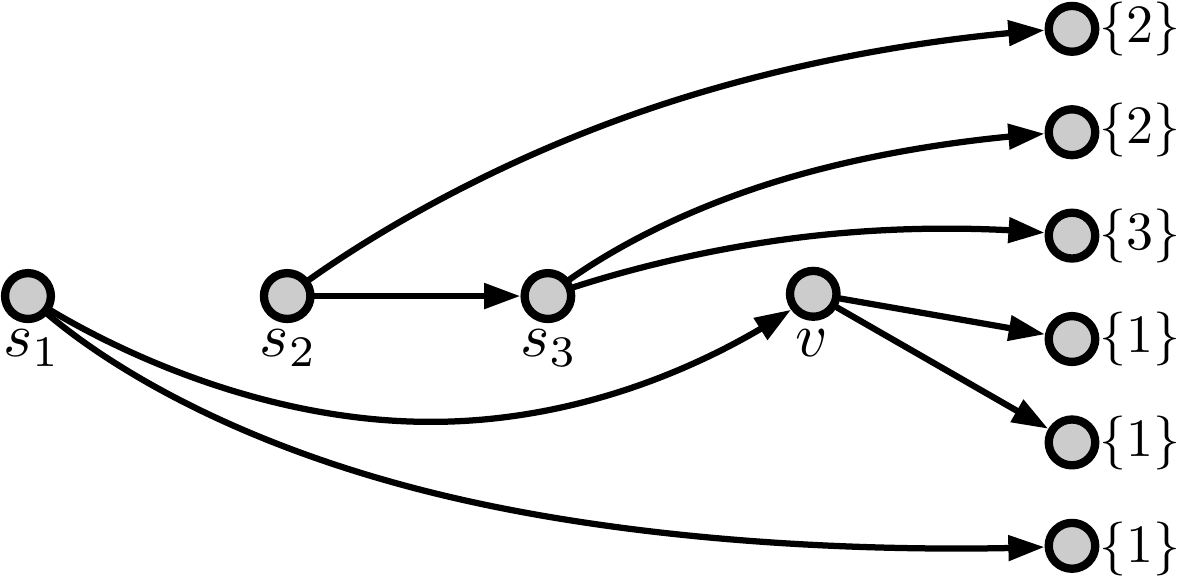}
    \caption{Solution with $\Delta^+(H) = 2$.}
    \label{subfig:broom_gapB}
  \end{subfigure}
  \caption{}
  \label{fig:broom_gap}
\end{figure}
\FloatBarrier
However, when $T = 2$ in the target algorithm, in the first step we generate the subproblem given in Figure \ref{fig:broom_gapC} by contracting node $v$ for the target $T = 2$. The resulting problem does \emph{not} have a feasible solution with max out-degree at most 2. Thus, the algorithm would not produce a feasible solution with max out-degree 2 for the original instance in Figure \ref{fig:broom_gap}.
\begin{figure}[h!]
	\begin{center}
		\scalebox{0.38}{
	\includegraphics{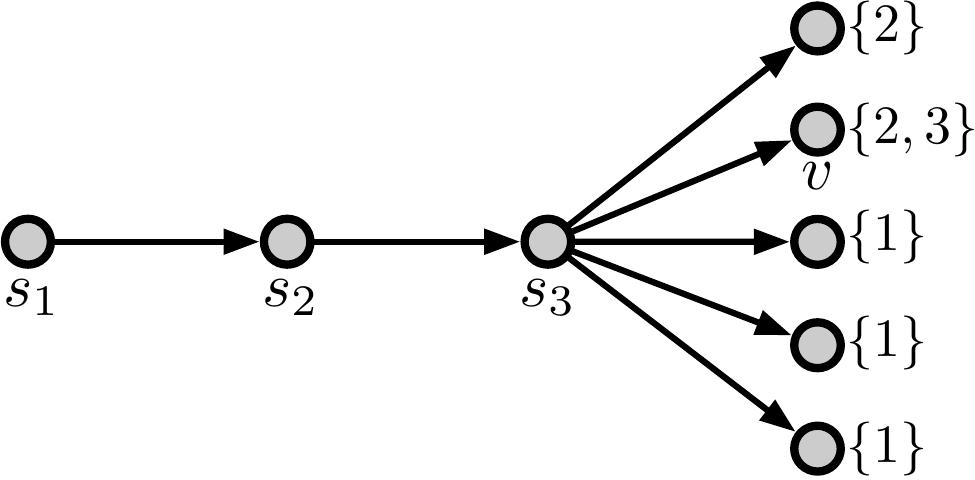}}
		\caption{}
		\label{fig:broom_gapC}
\end{center}
\end{figure} 
\FloatBarrier
It remains open whether or not out-tree instances are NP-hard. In Appendix \ref{app:out_tree_hardness}, we relate the hardness of these instances to that of a natural packing problem.

%% file: star.tex
\section{Arbitrary trees} \label{sec:arbitrary-trees}

In this section, we provide a framework for obtaining  approximation results for arbitrary tree instances. As a first step, we give an efficient 2-approximation when $D$ is a star. Our framework for obtaining an approximation algorithm for arbitrary tree instances builds on the approximability of junction tree instances, a generalization of stars. We also highlight the weakening strength of the witness set construction for more complex graphs.

\subsection{Star instances}
We have the following 2-approximate solution for star instances.
\begin{prop}
Given a star instance, $H = \{s_k t_k: k \in \K\}$ is a feasible solution with $\Delta^+(H) \leq 2 \Delta^*$.  
\end{prop}

\begin{proof}
Let $H = \{s_k t_k: k \in \K\}$. That is, $H$ is the subgraph obtained by routing all the commodities directly. The maximum out-degree of $H$, $\Delta^+(H)$, is equal $\max_{s \in \S}|\T(s)|$. Let $s^* = \argmax_{s \in \S} |\T(s)|$. If $s^* = v$, the center vertex, then~${W = \{v\}}$ and $\K' = \{k \in \K: s_k = v\}$, we obtain the following bound on~$\Delta^*$:
\[ \Delta^* \geq \LB(W, \K') = \left\lceil \frac{|\T(v)| + 1 - 1}{1} \right\rceil = |\T(v)|, \]
which matches the maximum out-degree of $H$. Alternatively, $s^* \neq v$. Then when~$W = \{s, v\}$ and $\K' = \{k \in \K: s_k = s^*\}$, we obtain the following bound on $\Delta^*$:
\[ \Delta^* \geq \LB(W, \K') =\left\lceil \frac{(|\T(s_i)| - 1) + 2 - 1}{2} \right\rceil = \left\lceil \frac{|\T(s_i)|}{2} \right\rceil.\]
Note that we have $|\T(s_i)| - 1$ since it could be that some commodity $k \in \K'$ has source $s^*$ and sink $v$, which falls within $W$. Thus, we have $\Delta_H \leq 2 \Delta^*$ as desired. 
\end{proof}

All of the algorithms presented rely on the lower bound formulation to certify the performance of the algorithm. However, there are instances where the best lower bound of the form provided in Lemma \ref{lemma:LB_multi} is at most $\frac23 \Delta^*$. Thus, arguing that the performance of an approximation algorithm for the star setting is better than a factor 3/2 cannot rely on the current lower bound construction. Future work is to find hardness results for the star setting, and to improve upon the factor 2 algorithm.

Consider the instance of \mmspp~in Figure \ref{fig:star_gap}, where there are $2n$ sinks ($\T$), and $m = {2n \choose n}$ sources ($\S$). Consider the $m$ subsets of $\T$ of size~$n$, and order the subsets arbitrarily as $T_1, T_2, \cdots, T_m$. The commodity set is constructed so that each source must route to one of the $m$ subsets of $n$ sinks. That is, ${\K = \{(s_i, t): i \in [m], t \in T_i\}}$. 

\begin{figure}[h!]
	\begin{center}
		\scalebox{0.5}{
			\includegraphics{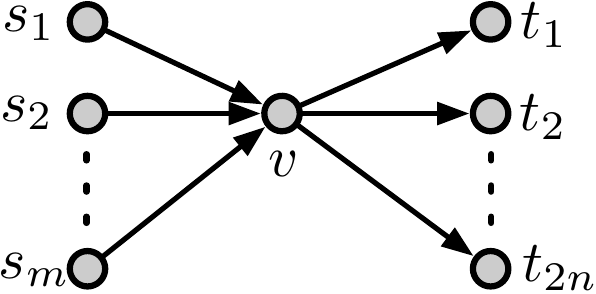}}
		\caption{}
		\label{fig:star_gap}
\end{center}
\end{figure} 
\FloatBarrier

Observe that $H = \{s_k t_k: k \in \K\}$ is one solution with $\Delta^+(H) = n$. Furthermore, if in any feasible solution $H$ we had that $\deg^+_H(v) \leq n-1$, then this would imply that $\deg_H^+(s) \geq n$ for some source $s$, since $v$ did not serve any of its corresponding sinks. Thus, $\Delta^* = n$. Now consider the family of lower bounds from Lemma  \ref{lemma:LB_multi}. It is not hard to see that the best lower bound possible is when~$W = \{s_i, s_j, v\}$ where $T_i \cap T_j = \emptyset$, and $\K'$ is the set of $2n$ commodities with sources in $\{s_i, s_j\}$. This gives a lower bound of $\left\lceil (2n + 1) / 3 \right \rceil$. 

\star*

%% file: tree_instances.tex
\subsection{Tree instances}
A tree instance of \mmspp~is a \emph{junction tree instance} if there is some node~$r$ in $D$ such that $r$ is a node in $P_k$ for all $k \in \K$. Suppose we had an algorithm for junction tree instances of \mmspp, $\A$, that returns feasible solutions with max out-degree at most a factor $\alpha$ larger than optimal. In this section, we show how an algorithm for junction tree instances can be used to obtain an approximation algorithm for tree instances. 

Let $\I = (D, \K)$ be a tree instance of \mmspp~and let $\Delta^*$ denote the minimum max out-degree of a feasible solution.  For any node $v$ in $D$, there is a set of commodities, $\K_v$, such that the dipath between the source and sink in~$D$,~$P_k$, includes $v$. That is,
\[ \K_v := \{k \in \K: v \in P_k\}.\]
For each of the remaining commodities not in $\K_v$, the path $P_k$ is fully contained in one of the components in $D - v$. Naturally, the union of a feasible solution to the instance $(D, \K_v)$ and the instance $(D, \K \setminus \K_v)$ is a feasible solution to the instance $\I$. Moreover, the two subproblems have a particular structure: $(D, \K_v)$ is a junction tree instance, and $(D, \K \setminus \K_v)$ is a set of tree instances -- one defined on each component of $D - v$. 

Let $v \in N$ and let $C$ be some connected component in $D-v$. Let  $\K(C)$ denote the set of commodities $k \in \K$ with $s_k, t_k \in V(C)$. That is,
\[ \K(C) := \{k \in \K: s_k, t_k \in V(C)\}. \]
Note that each of the tree instances in $(D, \K \setminus \K_v)$ is the instance~$(C, \K(C))$ for some component $C \in D - v$, and $C$ has fewer nodes than $D$. For each instance~$(C, \K(C))$, we again decompose the instance into a junction tree instance and a set of smaller tree instances. We repeat this process until all sub-instances of $\I$ are either junction tree instances, or defined on digraphs with a single node. 

Suppose each node in $D$ is in at most $\beta$ of the junction tree sub-instances of $\I$. Then, applying the algorithm $\A$ to each of the junction tree instances and taking the union of the solutions gives a feasible solution for $\I$, with max out-degree at most $\alpha \cdot \beta$ times $\Delta^*$. It is not hard to show that the node $v$ in each iteration can be chosen so that $\beta = \log(n)$. Specifically, we select the node $v$ so that each connected component in $D - v$ has at most $|N|/2$ nodes.

We define the following algorithm for tree instances.
\begin{algorithm}[h!]
\setcounter{AlgoLine}{0}
	\DontPrintSemicolon
	\KwIn{A tree instance $\I = (D, \K)$ of \mmspp, where $D = (N, A)$, and an $\alpha$-approximation algorithm, $\A$, for junction tree instances.}
	\If{$|N| = 1$}{
        \Return $(N, \emptyset)$
	}
    Let $v$ be a node in $N$ such that each connected component in $D - v$ contains at most $|N|/2$ nodes.\\
    $H_v \leftarrow \A((D, \K_v))$\\
    Let $\C = \{C_1, C_2, \ldots, C_{\ell}\}$ be the set of connected components in $D - v$.\\
    \For{$C \in \C$}{
        $H_v^C \leftarrow \mathtt{tree}$-$\mathtt{alg}((C, \K(C))$
    }
    $H \leftarrow H_v \bigcup_{C \in \C} H_v^C$\\
    \Return $H$
	\caption{$\mathtt{tree}$-$\mathtt{alg}(\I, \A)$}\label{alg:trees}
\end{algorithm}
\FloatBarrier

\begin{prop}
Let $\I = (D, \K)$ be a tree instance of \mmspp. Given an $\alpha$-approximation algorithm for junction tree instances, Algorithm \ref{alg:trees} returns a solution $H$ for $\I$ with $\Delta^+(H) \leq \alpha \log(n) \Delta^* $
\end{prop}

\begin{proof}
We first argue that the solution returned is feasible. Let $k \in \K$. Observe that if $k \notin \K_v$, then $P_k$ must be fully contained in a component $C$ in $D-v$. Thus, at any point in the algorithm when a tree instance is decomposed into a junction tree instance and a set of smaller tree instances, one of the instances contains the path $P_k$. 

Furthermore, it is clear that the minimum max out-degree of any subproblem~$\I' = (D', \K')$ generated by the algorithm provides a lower bound on $\Delta^*$, since any solution to $\I$ must also contain a subgraph $H' \subseteq \cl(D')$ that is a solution for $\I'$. Finally, each node is contained in at most $\log (n)$ of the subproblems on which algorithm $\A$ is executed. Therefore, $\Delta^+(H) \leq \alpha \log(n) \Delta^*$ as required. 
\end{proof}
Thus, we have proven the following Theorem. 
\junction*

%% file: conclusion.tex
\section{Conclusion}
\mnew{In this paper, we propose a mathematical model to determine an allocation of sort points that provides a feasible sortation plan. Even in the case where the underlying undirected physical network forms as tree, it is NP-hard to determine whether a feasible sort point allocation exists. We focus on the natural objective of minimizing the maximum number of sort points required at a warehouse, and define the directed min-degree problem of \mmspp.  

We present a simple and efficient combinatorial algorithm for solving single-source tree instances of \mmspp~with a quadratic speed-up over previous algorithms. We also present a fast combinatorial additive $1$-approximation algorithm for the out-tree setting. Moreover we prove that our analysis is tight by exhibiting an instance where the algorithm returns a solution that has max out-degree one greater than optimal. We also show that there is an inherent weakness of the family of lower bounds considered since there is a gap of one between the best lower bound and the minimum max out-degree for out-tree instances. For star instances of \mmspp, there are instances for which this gap is a multiplicative factor of $3/2$. An improvement to the approximation guarantee for the star setting, as well as a non-trivial approximation algorithm for junction tree instances remain open problems. A challenging open problem lies in finding good approximation algorithms for arbitrary graph structures. 

It remains open whether or not out-tree instances are NP-hard. In Appendix \ref{app:out_tree_hardness}, we relate the hardness of these instances to that of a natural packing problem. Furthermore, due to the limited understanding of optimizing for sort point problems, we also believe alternative objectives are of interest, such as varying degree constraints on nodes, and determining the minimum number of sort points required in any feasible solution. These settings are of high relevance in applications.}

\medskip 
\noindent {\bf Acknowledgements\ } We thank Sharat Ibrahimpur and Kostya Pashkovich for valuable discussions on lower bounds and the link to matroid intersection. 

%% file: Appendix/single_source.tex
\section{Single-source (missing proofs from Section \ref{subsec:single_out_trees})}\label{app:single_source}

In this section, we expand on the overview provided in Section \ref{subsec:single_out_trees}, and provide the full results for the single-source setting. 

We first show that this setting reduces to the problem of min-degree arborescence in a directed acyclic graph, which can be solved by matroid intersection in~$O(n^3 \log n)$ time \cite{schrijver2003combinatorial,ChakrabartyLS0W19} or a combinatorial algorithm in~$O(nm\log n)$ time~\cite{Yao}. We then present, in detail, the simple and fast combinatorial algorithm for single-source tree instances of \mmspp~that offers a quadratic speed-up over previous algorithms. While our algorithm has similarities to the one presented in \cite{Yao}, we can exploit the structure of transitive closure to reduce the number of arc swaps. We show that an efficient implementation of our approach gives a quadratic speed-up.

\begin{lemma}
Single-source tree instances of \mmspp~reduce to the problem of finding a min-degree arborescence in a directed acyclic graph. 
\end{lemma}
\begin{proof}
Let $\I = (D, \K)$ be a tree instance of \mmspp~with a single source,~$s$. Note that $D$ is an out-tree rooted at $s$, and $\cl(D)$ is a directed acyclic graph. We may assume that each leaf node in $D$ is a sink for some commodity, as otherwise this node could be removed. 

It suffices to prove that there is an optimal solution that contains an $s, v$-dipath for each node $v \in N$. Let $H$ be an arbitrary optimal solution. Certainly if~$v \in \T$, then an $s, v$-dipath is enforced by feasibility. So consider some node~${v \notin \T}$ and suppose there is no $s, v$-dipath in $H$. Then we may assume $v$ has no out-arcs in $H$ as otherwise they could be removed while maintaining feasibility. 

Let $T_v$ be the subtree in $D$ rooted at $v$. Since each leaf is a sink for some commodity, there is some arc $e = uw$ in $\delta^+_H(N \setminus V(\T_v))$. We obtain an optimal solution with an $s, v$-dipath by removing arc $uw$ and adding arcs $uv$ and $vw$. Repeating this process gives an optimal solution that is an $s$-arborescence.
\end{proof}

In a directed acyclic graph, the min-degree arborescence problem can be solved efficiently with matroid intersection for two partition matroids~\cite{schrijver2003combinatorial}. Given a bound $T$ on the maximum out-degree, the first partition matroid ensures each node has in-degree at most one, and the second ensures each node has out-degree at most $T$. We then run matroid intersection for $\log(n)$ many candidate values of $T$ to find an optimal solution. More precisely, for any candidate value $T$, we find a maximum-cardinality set (of edges) that is an independent set in both matroids. Since the base graph is acyclic, a maximum-cardinality independent set is an arborescence of degree bounded by $T$.

In the remainder of this section, we describe the simple combinatorial algorithm that solves single-source instances. First, we show that in the single-source setting, we can simplify the construction of witness sets. This set of lower bounds allows us to look at the base graph $D$ rather than keeping track of the commodity set. Specifically, we show that we can replace $|\K'|$ with $|\delta^+_D(W)|$ in Corollary~\ref{lemma:LB_multi} by recalling that all nodes $v$ in $D$ with $\delta_D^+(v) = \emptyset$ must be the sink of some commodity. 

\begin{corollary}\label{cor:LB_single_simple}
Suppose $\I = (D, \K)$ is a tree instance and $|\mathcal{S}| = 1$. Let $W \subseteq N$ such that $s \in W$ and $D[W]$ is a directed tree, and suppose $\delta^+_D(W) \neq \emptyset$. Then 
\[\Delta^* \geq \left\lceil \frac{|\delta^+_D(W)| + |W| - 1}{|W|} \right\rceil.\]
\end{corollary}
\begin{proof}
It suffices to show that for any such node set $W$, there is a set of commodities $\K'$ with cardinality $|\delta^+_D(W)|$ such that for each $k \in \K'$, $t_k \notin W$, and for all distinct $k, j \in \mathcal{K}'$, $\delta^+_{P_k}(W) \cap \delta^+_{P_j}(W) = \emptyset$. 

Consider an arc $e \in \delta^+_D(W)$. Since each node $v$ with $\deg^+_D(v) = 0$ is a sink for some commodity $k \in \K$, there must be a commodity $k$ such that $e \in P_k$, and more specifically, $e = \delta^+_{P_k}(W)$. For each arc $e \in \delta^+_D(W)$, let $k_e$ denote an arbitrary commodity such that $e = \delta^+_{P_{k_e}}(W)$, and let $\K' = \bigcup_{e \in \delta^+_D(W)} \{k_e\}$. Observe that for each pair of edges $e$ and $f$ in $\delta^+_D(W)$, $\delta^+_{P_{k_e}}(W) \cap \delta^+_{P_{k_{f}}}(W) = \emptyset$. Thus, $\K'$ satisfies the desired conditions. 
\end{proof}

Let $\I = (D, \K)$ be an instance of \mmspp. Recall that for each~${W \subseteq N}$ and~$\K' \subseteq \K$,
\begin{align*}
\mathtt{LB}_\I(W, \K') & := \begin{cases}
\left\lceil \frac{|\K'| + |W| - 1}{|W|} \right\rceil & \mbox{if $W, \K'$ satisfy the conditions of Lemma \ref{lemma:LB_multi} for $\I$}  \\
0 & \mbox{otherwise}
\end{cases} 
\end{align*}
For $W \subseteq N$ and $\K' \subseteq \K$, we define the following function. 
\begin{align*}
\mathtt{LB}_\I(W) & := \begin{cases}
\left\lceil \frac{|\delta^+_D(W)| + |W| - 1}{|W|} \right\rceil & \mbox{if $W$ satisfies the conditions of Corollary  \ref{cor:LB_single_simple} for $\I$}  \\
0 & \mbox{otherwise}
\end{cases}
\end{align*}
Again, when the instance $\I$ is clear from context, we drop the subscript. Observe that the two families of lower bounds have equal strength since for any choice of~${W \subseteq N}$, the subset $\K'$ that maximizes the value of $\mathtt{LB}(W, \K')$ has size $|\delta^+(W)|$. That is, 
\[ \max_{W \subseteq N, \K' \subseteq \K} \mathtt{LB}(W, \K') = \max_{W \subseteq N} \mathtt{LB}(W).\] 
We will prove in Proposition \ref{thm:single_source} that for single-source instances, $\Delta^*$ is equal to $\max_{W \subseteq N} \mathtt{LB}(W)$. Moreover, we will prove that $\argmax_{W \subseteq N} \mathtt{LB}(W)$ can be found in polynomial time as a byproduct of a exact polytime algorithm for \mmspp~in the single-source setting. 

\medskip \noindent {\bf The single-source algorithm\ }
In the single-source setting, $D$ is a directed out-tree rooted at the source $s$. The algorithm can be described as follows. We begin with the feasible subgraph $H_0$ set to $D$. Let $H_i$ denote the feasible subgraph obtained in iteration $i$. In iteration $i$, we identify a node in $H_{i-1}$ with the highest out-degree, denoted $v^*$. We then attempt to shift an arc $v^*w$ in $H_{i-1}$ to instead depart a node along the $s, v^*$-dipath in $D$ with out-degree at most $\Delta^+(H_{i-1}) - 2$ in $H_{i-1}$. If such a node exists, we let $u$ denote the nearest such node to $v^*$ in $D$, and define $H_i$ as the subgraph obtained from $H_{i-1}$ by replacing arc $v^* w$ with arc $u w$. If no such node exists, the algorithm terminates with $H = H_{i-1}$. This procedure is restated in Algorithm \ref{alg:single_source}, where $P_{uv}$ is used to denote the unique dipath in $D$ from $u$ to $v$. In the following algorithm, we define a \emph{topological ordering}, $\preceq$, as an ordering of the nodes so that for any arc $uv \in A$, $u \prec v$.

\begin{algorithm}[h!]\label{alg:single_source}
\setcounter{AlgoLine}{0}
	Assign a topological ordering $\preceq$ to the nodes \\
	$H_0 \leftarrow D$\\
	$i = 1$\\
	\While{True}{
	    Let $v^* \in \argmax_{v \in N} \deg^+_{H_{i-1}}(v)$\\
	    Let $v^* w \in \delta^+_{H_{i-1}}(v^*)$\\
	    $R \leftarrow \{u \in V(P_{s v^*}): \deg_{H_{i-1}}^+(u) \leq \deg_{H_{i-1}}^+(v^*) - 2\}$\\
	    
	    \If{$R \neq \emptyset$}{
	        Let $u \in R$ such that $y \preceq u$ for all $y \in R$\\
	        $H_i \leftarrow H_{i-1} \setminus \{v^* w\} \cup \{u w\}$}
	    \Else{
	        \Return $H = H_{i-1}$
	    }
		$i = i + 1$\\
}
	\caption{Local search algorithm for the single-source setting.}\label{alg:single_source}
\end{algorithm}
\FloatBarrier

Figure \ref{fig:app_singlesource_ex} demonstrates the steps of the algorithm when the base graph $D$ is as provided in Figure \ref{subfig:app_a}, and $\T$ is the set of leaves. The square node is the selected max out-degree node in each iteration. Observe that in iteration 3, the set $R$ is empty and so the algorithm terminates. 

\begin{figure}[h!]
  \begin{subfigure}{0.3\textwidth}
    \centering
    \includegraphics[width=\textwidth]{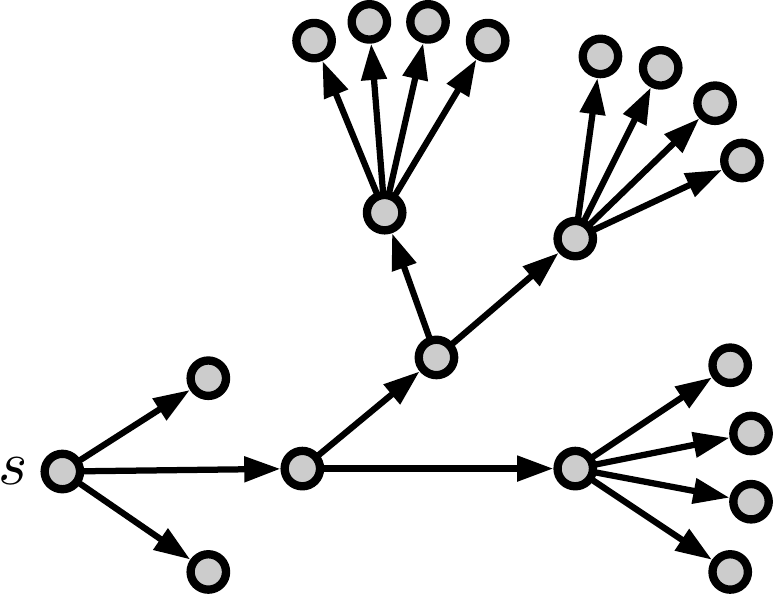}
    \caption{$H = D$}
    \label{subfig:app_a}
  \end{subfigure}
  \hfill
  \begin{subfigure}{0.3\textwidth}
    \centering
    \includegraphics[width=\textwidth]{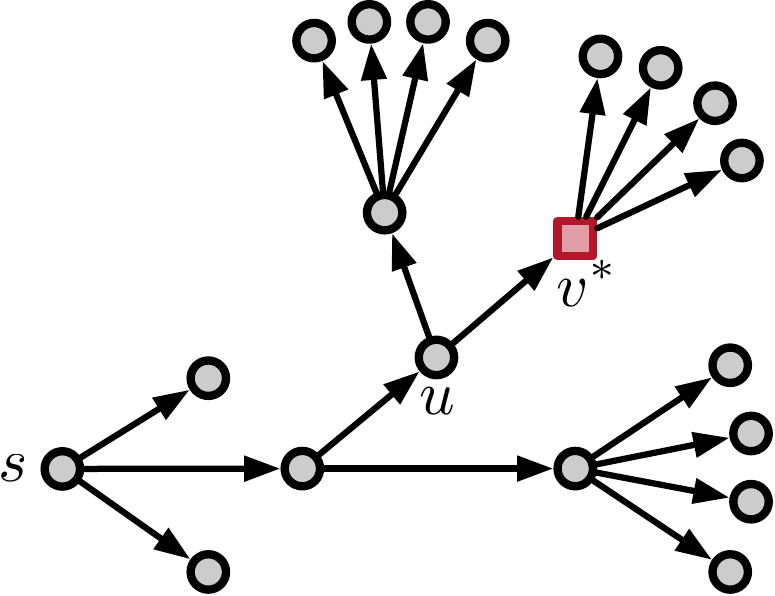}
    \caption{Iter. 1: $v^*$ identified.}
    \label{subfig:b}
  \end{subfigure}
  \hfill
  \begin{subfigure}{0.3\textwidth}
    \centering
    \includegraphics[width=\textwidth]{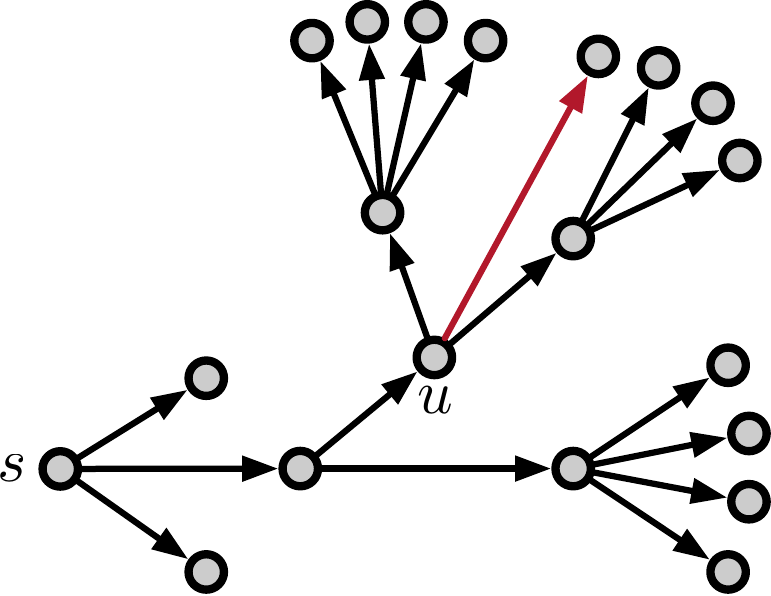}
    \caption{Iter. 1: $v^*w$ shifted.}
    \label{subfig:c}
  \end{subfigure}

  \medskip

  \begin{subfigure}{0.3\textwidth}
    \centering
    \includegraphics[width=\textwidth]{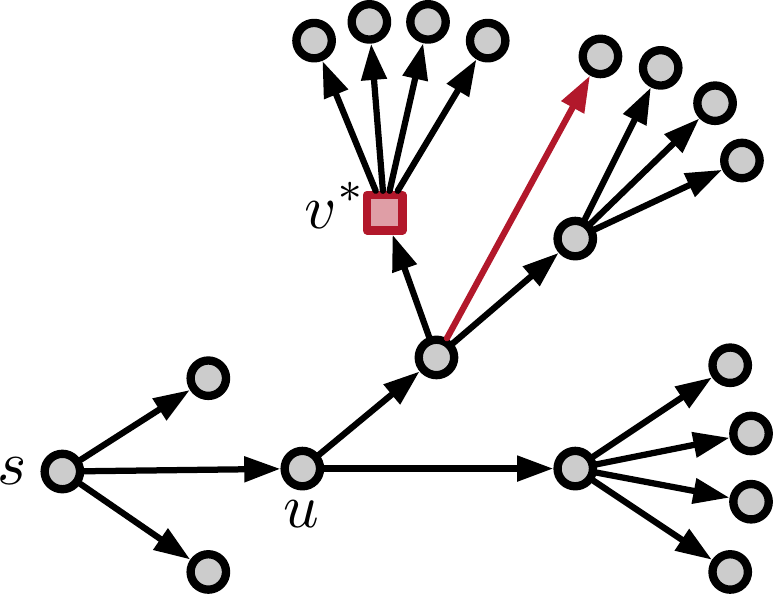}
    \caption{Iter. 2: $v^*$ identified.}
    \label{subfig:d}
  \end{subfigure}
  \hfill
  \begin{subfigure}{0.3\textwidth}
    \centering
    \includegraphics[width=\textwidth]{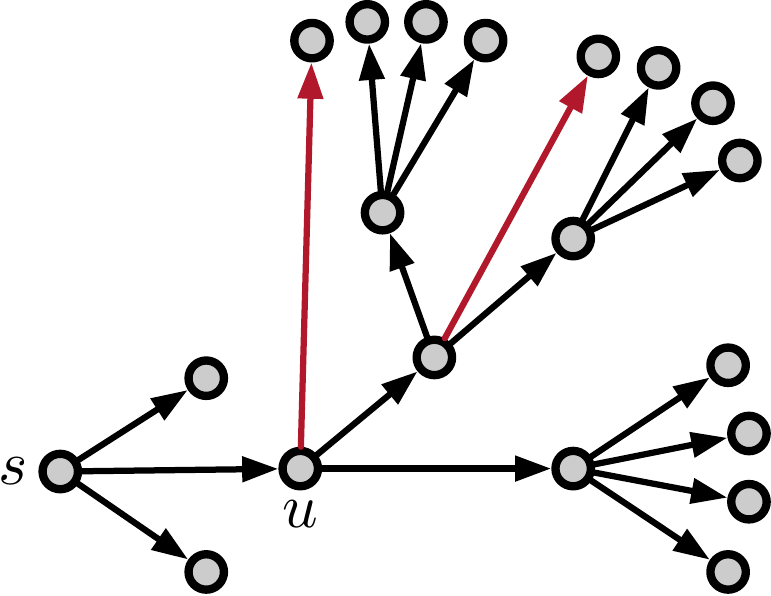}
    \caption{Iter. 2: $v^*w$ shifted.}
    \label{subfig:e}
  \end{subfigure}
  \hfill
  \begin{subfigure}{0.3\textwidth}
    \centering
    \includegraphics[width=\textwidth]{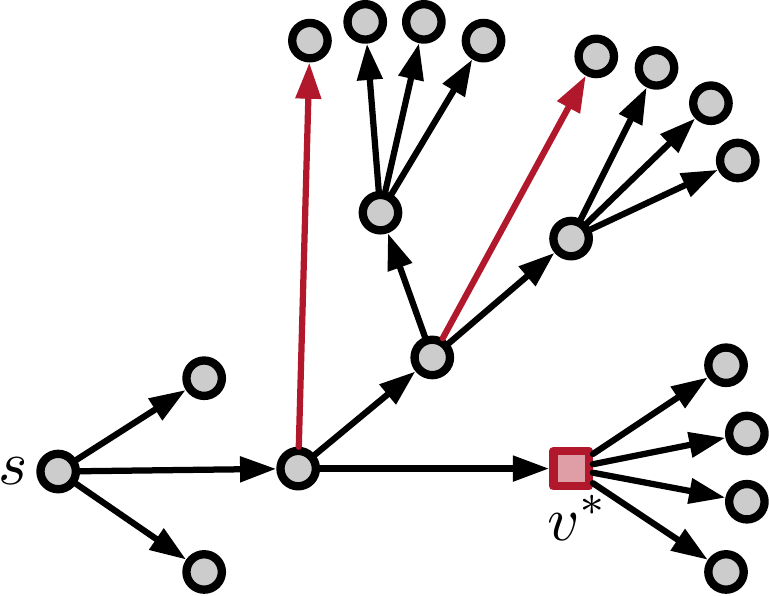}
    \caption{Iter. 3: $v^*$ identified.}
    \label{subfig:f}
  \end{subfigure}
  \caption{Execution of local search algorithm for single-source setting.}
  \label{fig:app_singlesource_ex}
\end{figure}
\FloatBarrier

The following observations are used in arguing for the efficiency and optimality of the local search algorithm.
\begin{enumerate}
	\item[(P1)] If a vertex $u$ has $\deg_{H_i}^+(u) \geq \Delta^+(H_i) - 1$ in iteration $i$, then in each subsequent iteration $j \geq i$, $\deg_{H_j}^+(u) \geq \Delta^+(H_j) - 1$;
	\item[(P2)] Subgraph $H_i$ is an $s$-arborescence in each iteration.
\end{enumerate}
The first statement holds since in each iteration, the only node that decreases in out-degree is the identified max out-degree node. Roughly, (P1) states that once a node has high out-degree relative to other nodes in some iteration, it may decrease in out-degree, but it remains a high out-degree node relative to the other nodes in each iteration. The second observation follows from the fact that~$D$ is directed tree, and in each iteration $H$ remains a single connected component without increasing the number of arcs.

\begin{prop}\label{thm:single_source}
Algorithm \ref{alg:single_source} returns a feasible subgraph $H$ with $\Delta^+(H) = \Delta^*$ given an instance with a single source in $O(n^2 \log n)$ time.
\end{prop}

\begin{proof}
First, we argue that the algorithm terminates in polynomial time. 

In each iteration, the out-degree of some maximum out-degree vertex is reduced by 1. Given a maximum out-degree of \(\Delta\), there are at most \(\lfloor \frac{n}{\Delta} \rfloor\) iterations until the maximum out-degree is reduced to \(\Delta-1\).
Since it holds \(\Delta\in[2,n-1]\) (unless the initial maximum out-degree is equal to 1), the total number of iterations is upper bounded by \[\sum\limits_{\Delta=2}^{n-1} \left\lfloor \frac{n-1}{\Delta} \right\rfloor = O(n\log n).\]
Since each iteration can be implemented in $O(n)$ time, the algorithm runs in~$O(n^2 \log n)$ time.

It remains to argue that the subgraph $H$ produced by the algorithm is both feasible and $\Delta^+(H) = \Delta^*$. The feasibility of $H$ follows from the observation that in each iteration, for each node $v \in N$, the current subgraph contains an~$s, v$-dipath. 

We now prove that $\Delta^+(H) = \Delta^*$, by presenting a witness set $W$ such that $\mathtt{LB}(W) = \Delta^+(H)$. Let $v^*$ be a max out-degree node in $H$. We iteratively construct $W$ as follows. First, $W = V(P_{sv^*})$. Note that for all $v \in V(P_{sv^*})$, $\deg_{H}^+(v) \geq \Delta^+(H)-1$ since the algorithm terminated. Then, while there is a node $u \notin W$ such that in some iteration an arc $uw$ was exchanged for an arc~$vw$ where $v \in W$, we add $V(P_{su})$ to $W$. Note that if such an exchange of arcs occurred, $u$ was a max out-degree node in some iteration $i$, and all nodes along the $v, u$-dipath, excluding $v$, had out-degree at least~$\Delta^+(H_i) - 1$ in $H_i$. By observation $(P1)$, it follows that each such node has out-degree at least~$\Delta^+(H) - 1$ in~$H$. Thus, for each node $v \in W$, $\deg^+_H(v) \geq \Delta^+(H) - 1$. Figure \ref{fig:certificate} shows the node set $W$ obtained for the instance solving in Figure \ref{fig:app_singlesource_ex}. 

\begin{figure}[ht]
	\begin{center}
		\scalebox{0.36}{
        \includegraphics{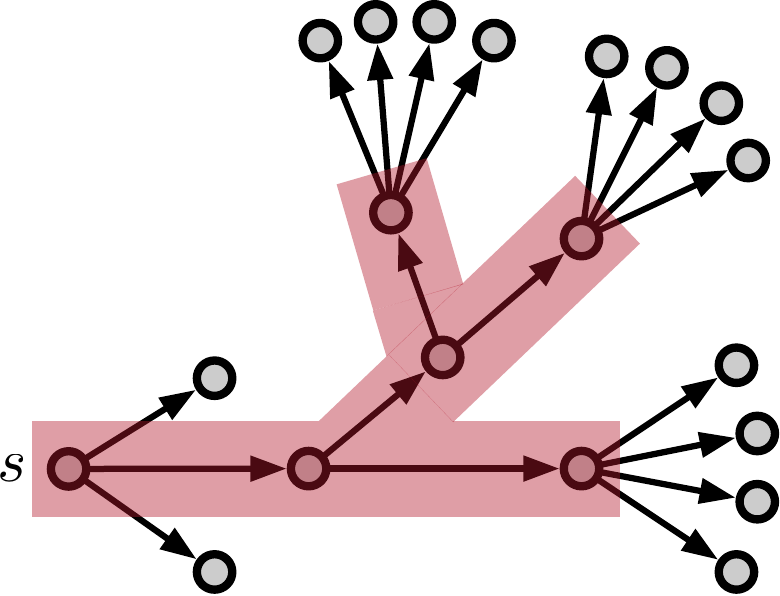}}
		\caption{The set $W$ certifies optimality of $H$ from Figure \ref{subfig:f}.}
		\label{fig:certificate}
\end{center}
\vspace{-.75cm}
\end{figure}
\FloatBarrier

Observe that $D[W]$ is connected since $W$ is formed by adding the node set of dipaths from $s$ in $D$. As a result, 
\[|\delta^+_{H}(W)| \geq (\Delta^+(H) - 1)|W| + 1 - (|W| - 1) = (\Delta^+(H) - 2)|W| + 2.\] 
Furthermore, $W$ contains $s$ and $\delta^+_D(W) \neq \emptyset$, so $W$ satisfies the conditions of Corollary \ref{cor:LB_single_simple}. Thus, it suffices to show that $|\delta^+_{D}(W)| \geq |\delta^+_{H}(W)|$, since then 
\[ \mathtt{LB}(W) = \ceil*{\frac{|\delta^+_D(W)| + |W| - 1}{|W|}} \geq \ceil*{\frac{(\Delta^+(H) - 2)|W| + 2 + |W| - 1}{|W|}} \geq \Delta^+(H)\]
Let $vw \in \delta^+_{H}(W)$. This arc corresponds to the unique arc entering $w$ in $D$, i.e.,~$uw$ (where possibly $u = v$). By the iterative construction of $W$, $u \in W$, and so~$uw \in \delta^+_D(W)$. Furthermore, no two arcs in $\delta^{+}_H(W)$ correspond to the same arc in~$\delta^{+}_D(W)$ since throughout the algorithm, the feasible subgraph remains an arborescence (and hence the total number of arcs is unchanged). If a pair of arcs in $\delta^{+}_H(W)$ were obtained via exchanges of arcs originating at the same arc~$uw$ in~$\delta^{+}_D(W)$, then at some point in the algorithm the number of arcs would have decreased. Thus, $|\delta^+_{D}(W)| \geq |\delta^+_{H}(W)|$ as required. 
\end{proof}

To obtain a more efficient algorithm for the single-source setting, we can reduce the number of operations by assuming we are given a target max out-degree, $T$. Consider a node $v$ with \emph{more} than $T$ nodes in $N^+_D(v)$, and let~${\alpha_v = |N^+_D(v)| - T}$ denote the number of excess nodes. If there is a solution with~$\Delta^+(H) \leq T$, at least $\alpha_v$ nodes in the set $N^+_D(v)$ must instead be reached by arcs departing a predecessor of $v$ in $D$. We divide the set $N^+_D(v)$ into a set of \emph{fixed nodes}, $F_v$, of cardinality $\min\{ |N_D^+(v)|, T\}$, and the \emph{remaining nodes} $R_v = N_D^+(v) \setminus F_v$. We then generate a solution $H$, if $T \geq \Delta^*$, so that $H$ contains an arc $v w$ for all $w \in F_v$, and an arc $u w$ for all $w \in R_v$, for some predecessor of $v$ in $D$. 

We store the set of descendant nodes of a node $v$ that have not yet been allocated in a set $D_v$, which is built throughout the algorithm. In other words, $D_v$ is the set of remaining nodes that $v$ inherited from its descendants that have not yet been allocated to a node. Moving from the leaves towards the root, $s$, we shift any remaining nodes in $R_v$ and $D_v$, to the set $D_{p(v)}$, where $p(v)$ denotes the parent of $v$ in $D$. If we encounter a node with \emph{fewer} than $T$ nodes in $F_v$, we then take a maximum of $-\alpha_v$ nodes from the set $D_v$ and allocate them to $v$. Note that $-\alpha_v = T - |N^+_D(v)| > 0$ and $R_v = \emptyset$ for such nodes. Since the arcs departing the max out-degree nodes in Algorithm \ref{alg:single_source} were chosen arbitrarily, the correctness of the algorithm implies that this procedure will produce a feasible solution $H$ with~${\Delta^+(H) \leq T}$, so long as $T \geq \Delta^*$. For the root node in $D$, $s$, let $p(s):= s$. 

\begin{algorithm}[h!]
\setcounter{AlgoLine}{0}
\KwIn{A target $T \in \mathbb{Z}_{\geq 1}$, and a single-source instance of \mmspp, $\I = (D, \K)$, where $D=(N,A)$ with root $s$.}
    $V \leftarrow \{s=v_1, v_2, \ldots, v_n\}$, in topological order \\
    $F_v \leftarrow \min\{|N_D^+(v)|, T\}$ nodes from $N_D^+(v)~\forall  v \in N$\\
    $R_v \leftarrow N_D^+(v) \setminus F_v ~\forall v \in N$ \\
    $D_v \leftarrow \emptyset ~\forall v \in N$ \\
    $\alpha_v = |N_D^+(v)| - T$ for all $v \in N$ \\
    \For{$i = n$ to $1$}{
        \If{$\alpha_v < 0$}{
            Let $U$ be a set of $\min\{|D_v|, -\alpha_v\}$ nodes from  $D_v$\\
            $D_v \leftarrow D_v \setminus U$\\
            $F_v \leftarrow F_v \cup U$\\
        }
        $D_{p(v)} \leftarrow D_{p(v)} \cup R_v \cup D_v$\\
    }
    \If{$D_s \neq \emptyset$}{
		\Return $\emptyset$ \\}
    \Else{
        \Return $H = \{vw: w \in F_v, v \in N\}$
    }
	\caption{Target algorithm for the single-source setting.}
 \label{alg:single_source_target}
\end{algorithm}
\FloatBarrier

\begin{prop}
Algorithm \ref{alg:single_source_target} runs in $O(n \log n)$ time, and returns a solution $H$ with $\Delta^+(H) \leq T$ whenever $T \geq \Delta^*$.
\end{prop}

\begin{proof}
The correctness of the algorithm follows directly from that of Algorithm \ref{alg:single_source}. 

In the first step, the topological ordering can be computed in $O(n)$ time. In this process, for each node, we create an object that stores the parent node. The sets of nodes $F_v, R_v$, and $D_v$ as well as the value of $\alpha_v$ are computed in $O(n)$ time. Producing $H$ given the sets $F_v$ also takes $O(n)$ time. It remains to argue that the for-loop can be executed in $O(n \log n)$ operations. 

This can be argued via {\em union-find} type arguments (e.g., see Lecture 10 in \cite{kozen}). We store the sets $F_v, R_v$, and $D_v$ as linked lists, with each element storing both a pointer to the next element, as well as a pointer to the head of its list. The union of two lists of length $L_A$ and $L_B$ takes $O(\min(L_A, L_B))$ time, by splicing the shorter list into the longer list. Then, each element $x$ has its head pointer updated at most $\log n$ times. 

Our algorithm has the added complication that elements are \emph{removed}. Since we remove an arbitrary set of nodes for a fixed size, $-\alpha$, this can be done in $O(-\alpha)$, since we can select the first $-\alpha$ nodes after the head of the list (or the entire list). Furthermore, each node is deleted at most once, and so the total operations for deletion is $O(n)$. 

To maintain that element $x$ has its head pointer updated at most $\log n$ times in the union operations, instead of splicing the shorter of the two sets into the larger one, we measure the lengths of the sets \emph{as though no deletions have occurred}. Therefore the runtime of the algorithm is $O(n \log n)$ as claimed. 
\end{proof}

\singlesource*
\begin{proof}
By executing Algorithm \ref{alg:single_source_target} for $\log n$ many possible target values, we determine the smallest value of $T$ for which Algorithm \ref{alg:single_source_target} gives a feasible subgraph~$H$ with~$\Delta^+(H) \leq T$. Since Algorithm \ref{alg:single_source_target} did not return a feasible subgraph for the target $T - 1$, $\Delta^* > T$. Therefore, $\Delta^+(H) = \Delta^*$.  
\end{proof}

%% file: Appendix/out_tree_hardness.tex
\section{Out-tree hardness}\label{app:out_tree_hardness}

It remains open whether or not out-tree instances of \mmspp~are NP-hard. Establishing hardness for this setting would demonstrate that our additive 1-approximation algorithm was the best possible in Section \ref{sec:out_trees}. Even in the seemingly simple case of \emph{broom instances} -- instances in which the underlying undirected graph of $D$ is a broom -- appears challenging. In this section we demonstrate that broom instances are at least as hard as a particular packing problem, so long as its inputs are polynomially bounded. We believe this packing problem is of independent interest, and establishing it is strongly NP-hard is one avenue to proving that \mmspp~is NP-hard on brooms, and out-trees more generally. 

$\mathtt{SIGNED}$-$\mathtt{VALUES}$-$\mathtt{BIN}$-$\mathtt{PACKING}$~($\mathtt{SV}$-$\mathtt{BP}$): Given a set $\W = \{w_1, w_2, \ldots, w_n\}$ of~$n$ integers, determine if there is a partition of $[n]$ into sets $N_1, N_2, \ldots, N_\ell$ for some~$\ell \in [n]$ such that for each part $i \in [\ell]$, $\sum_{j \in N_i} w_j \leq 1$. 

Note that we may assume that none of the input integers are 0 or 1, since $\W$ is a $\mathtt{YES}$-instance if and only if $\W' = \{w_i: w_i \in \W, w_i \notin \{0,1\}\}$ is a $\mathtt{YES}$-instance. In the proof of Lemma \ref{lemma:out_tree_hardness}, we establish that broom instances of \mmspp~are at least as hard as the problem of $\mathtt{SV}$-$\mathtt{BP}$ when the integers are bounded by a polynomial in $n$. Thus, establishing that $\mathtt{SV}$-$\mathtt{BP}$ is strongly NP-hard would prove that \mmspp~on brooms (and thus out-trees) is NP-hard. 

\begin{open}
Is $\mathtt{SV}$-$\mathtt{BP}$ strongly NP-hard?
\end{open}

\begin{lemma}\label{lemma:out_tree_hardness}
Broom instances of \mmspp~are at least as hard as $\mathtt{SV}$-$\mathtt{BP}$~with weights that are polynomial in $n$.
\end{lemma}

\begin{proof}
We show that there is a reduction from $\mathtt{SV}$-$\mathtt{BP}$ to broom instances of \mmspp, when the integers are polynomial in $n$. Let $\W = \{w_1, w_2, \ldots, w_n\}$ be the input integers of an instance of $\mathtt{SV}$-$\mathtt{BP}$, where the integers are polynomial in $n$, $w_1 \geq w_2 \geq \ldots, w_n$. We say that a partition of $\W$ is \emph{valid} if each part sums to at most 1. 

Let $T \geq 2$. We will construct a broom instance of \mmspp, with sets of sources $S_1, S_2, \ldots, S_n$, so that there is solution with max out-degree $T$ if and only if the weakly-connected sets of sources in an optimal solution corresponds to a valid partition of the index set $[n]$. We construct a corresponding instance of the broom problem with target $T$ as follows. 

\textbf{Construction of broom instance, given $T \geq 2$}: we create a set of sources, $S_i$, for each index $i \in [n]$. In this set, we generate $\alpha_i + 1$ sources, where~$\alpha_i$ is the smallest positive integer such that $w_i + \alpha_i(T - 1) \geq 0$. Let $\beta_i := w_i + \alpha_i(T - 1)$. We introduce $\beta_i + T - 1$ sinks, $T_i = \{t_i^1, t_i^2, \ldots, t_i^{\beta_i + T - 1}\}$. 

We construct a set $\K_i$ of commodities with sources in $S_i$ and sinks in $T_i$ as follows. The first source, $s^1_i$, is assigned $\beta_i$ sinks. The final source, $s_i^{\alpha_i + 1}$ is assigned the remaining $T-1$ sinks. Finally, we ensure that each source in $S_i$ must have a path to a common sink node by defining a commodity with source $s$ and sink $t_i^{\beta_i + 1 - 1}$ for each source $s \in S_i$. Specifically, we have the following commodities. 
\begin{itemize}
\item[$\bullet$] $\{(s_i^1, t_i^j): j \in [\beta_i]\}$
\item[$\bullet$] $\{(s_i^{\alpha_i + 1}, t_i^j): j \in [\beta_i + T - 1] \setminus [\beta_i] \}$
\item[$\bullet$] $\{(s_i^{j}, t_i^{\beta_i + T - 1}): j \in [\alpha_i]\}$
\end{itemize}
Note that each node in $\{s_i^2, \ldots, s_i^{\beta_i + T - 2}\}$ is the source of only a single commodity. 

The corresponding broom instance $\I = (D, \K)$ has the digraph $D$ as depicted in Figure \ref{fig:num_hardness}, with a dipath on the source sets $S_1$ to $S_n$, and the sinks are the leaves. The commodity set is $\K = \cup_{i \in [n]} \K_i$. 

\begin{figure}[ht]
	\begin{center}
		\scalebox{0.4}{\includegraphics{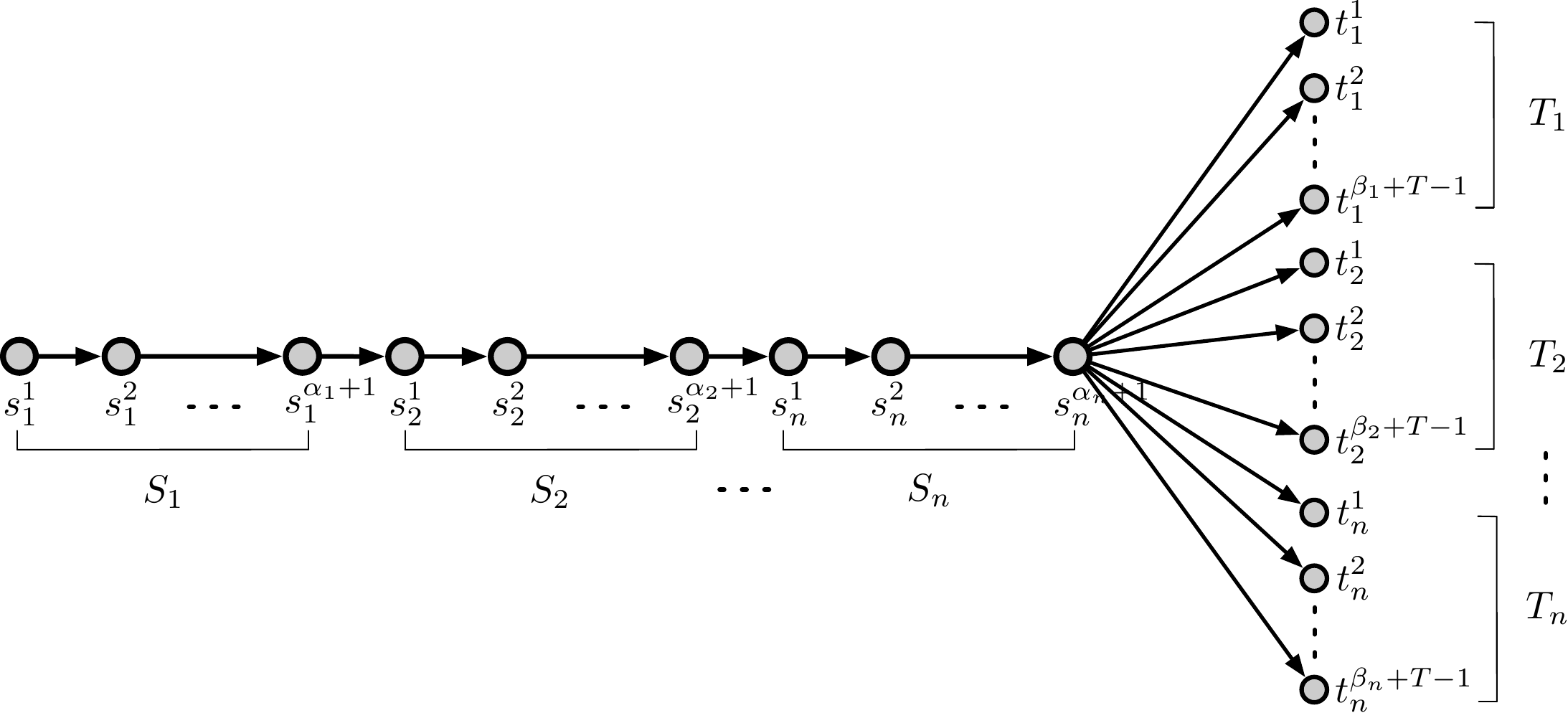}}
		\caption{Digraph $D$ for the broom instance.}
		\label{fig:num_hardness}
\end{center}
\end{figure}
\FloatBarrier

We now prove that $\Delta^* \leq T$ if and only if there is a valid partition of $\W$. 

\noindent$(\Rightarrow)$ Suppose there is a solution, $H$, to the instance $\I = (D, \K)$ with~${\Delta^+(H) \leq T}$. The solution splits into weakly-connected components, $H_1, H_2, \ldots, H_{\ell}$. Recall that each set of sources $S_i$ will be in the same weakly-connected component due to the shared sink, $t_i^{\beta_i + T - 1}$. For each component $j \in [\ell]$, let~${N_j = \{i: S_i \in H_j\}}$. We will prove that $N_1, N_2, \ldots, N_{\ell}$ is a valid partition of $\W$. 

Let $\S_j$ and $\T_j$ denote the set of sources and sinks in component $H_j$ respectively. By definition of $N_j$, it follows that $\S_j = \bigcup_{i \in N_j} S_i$ and $\T_j = \bigcup_{i \in N_j} T_i$. Furthermore, the node set of $H_j$ is $\S_j \cup \T_j$. The fewest arcs possible for the subgraph $H_j$, since it is weakly-connected, is $|\S_j| + |\T_j| - 1$. Since every arc in~$H_i$ departs a node in $\S_j$, $\Delta^+(H_j) \leq T$ implies that $\frac{|\S_j| + |\T_j| - 1}{|\S_j|} \leq T$. Therefore, 
\begin{align*}
& |\S_j| + |\T_j| - 1 \leq T |\S_j| \\
\iff & \sum_{i \in N_j} (\alpha_i + 1) + \sum_{i \in N_j} (\beta_i + T - 1) - 1 \leq T \sum_{i \in N_j} (\alpha_i + 1) \\
\iff & \sum_{i \in N_j} \alpha_i + |N_j| + \sum_{i \in N_j} (w_i + \alpha_i(T - 1)) + (T - 1) |N_j| - 1 \leq T \sum_{i \in N_j} \alpha_i + T |N_j|\\
\iff & \sum_{i \in N_j} \alpha_i + \sum_{i \in N_j} (w_i + \alpha_i(T - 1)) - 1 \leq T \sum_{i \in N_j} \alpha_i \\
\iff & \sum_{i \in N_j} \alpha_i + \sum_{i \in N_j} w_i + T \sum_{i \in N_j} \alpha_i - \sum_{i \in N_j} \alpha_i - 1 \leq T \sum_{i \in N_j} \alpha_i\\
\iff & \sum_{i \in N_j} w_i \leq 1.
\end{align*}
Since this holds for each $j \in [\ell]$, the partition $N_1, N_2, \ldots, N_{\ell}$ satisfies~${\sum_{i \in N_j} w_i \leq 1}$ for each $j \in [\ell]$. 

\noindent$(\Leftarrow)$ Suppose there is a valid partition, $N_1, N_2, \ldots, N_{\ell}$, for $\W$. We will form a solution $H$ to the broom instance $\I = (D, \K)$, consisting of $\ell$ connected components,~$H_1, H_2, \ldots, H_{\ell}$, where each component has maximum out-degree at most $T$. 

\noindent\textbf{Construction of $H_j$}. Let $j \in [\ell]$ and let $\S_j = \bigcup_{i \in N_j} S_i$ and $\T_j = \bigcup_{i \in N_j} T_i$. First, $H_j$ is formed by taking the unique dipath in $\cl(D)$ on the nodes in $\S_j$. Let $\P_j$ denote this path. It remains to add arcs to each of the sinks in $\T_j$ so that there is an $s_k, t_k$-dipath for each commodity $k \in \K_i$, for each $i \in N_j$, while ensuring that $\Delta^+(H_j) \leq T$.

Let $\{s_1, s_2, \ldots, s_{|\S_j|}\}$ be the set of sources in $\S_j$, ordered so that there is an~${s_i, s_j}$-dipath in $D$ whenever $i \leq j$. Let $\{t_1, t_2, \ldots, t_{|\T_j|}\}$ be the set of sinks in~$\T_j$, ordered first by increasing subscript, and then increasing superscript (top to bottom from Figure \ref{fig:num_hardness}). 

We add arcs to $H_j$ from $\S_j$ to $\T_j$ as follows. First we add an arc from $s_{|\S_j|}$ to each of the final $T$ sinks in $\T_j$. In decreasing order of sources, we continue to add arcs from the current sink to the last $T-1$ sinks remaining in $\T_j$, until no sinks remain, or we run out of sources. Observe that the resulting graph, $H_j$, has $\Delta^+(H_j) \leq T$. 

If $|\T_j| \leq (T-1) |\S_j| + 1$, then $H_j$ will include all sinks in $\T_j$. This is indeed the case, since 
\begin{align*}
|\T_j| \leq (T-1) |\S_j| + 1  \iff |\T_j| + |\S_j| - 1 \leq T |\S_j| \iff \sum_{i \in N_j} w_i \leq  1, 
\end{align*}
as proven in the previous direction.

Thus, it remains to show that for each commodity $k \in \K$ with $s_k \in \S_j$ and~$t_k \in \T_j$, there is an $s_k, t_k$-dipath in $H_j$. We prove the following claims, and recall that we may assume all weights $w_i$ are not equal to 0 or 1. 

\noindent\textbf{Claim 1:} If $w_i \geq 2$, then $|T_i| \geq |S_i|T$. 
\par
\leftskip=1cm
\noindent\emph{Proof.} If $w_i \geq 2$, then $\alpha_i = 1$ and $|S_i| = 2$. Furthermore, it holds that~${\beta_i = w_i + (T-1) \geq T + 1}$, and so $|T_i| = \beta_i + T - 1 \geq 2T$. Therefore,~$|T_i| \geq |S_i| T$. 
\par\leftskip=0cm
\noindent\textbf{Claim 2:} If $w_i \leq 0$, then $|T_i| \leq |S_i|(T-1)$.
\par\leftskip=1cm
\noindent \emph{Proof.} If $-(T-1) \leq w_i \leq 0$, then $\alpha_i = 1$ and $|S_i| = 2$. Furthermore, it must be that $0 \leq \beta_i \leq T-1$ and so $|T_i| = \beta_i + T - 1 \leq 2T-2$. Thus, it is $|T_i| \leq |S_i| (T - 1)$. 

If instead $w_i < -(T-1)$, then $\alpha_i \geq 2$, and $|S_i| \geq 3$. Furthermore, it holds~$\beta_i \leq T$, and so $|T_i| \leq 2T-1$. Note that for all $T \geq 0$, we have~$2T-1 \leq 3(T-1)$. Thus, $
|T_i| \leq 3(T-1) \leq |S_i|(T - 1)$. 
\par\leftskip=0cm

For a contradiction, suppose there is some commodity with source $s_p$ and sink~$t_q$ (indices from ordered nodes in $\S_j$ and $\T_j$), such that $H_j$ contains no~${s_p, t_q}$-dipath. Moreover, among all commodities of this form, assume that $p$ is the largest possible index. Since there is no $s_p, t_q$-dipath, the arc from $\S_j$ to $t_q$ departs a node $s_r$ with $r < p$. Similarly, since there is no $s_p, t_q$-dipath in $H_j$, it follows that 
\begin{align}\label{ineq1}
    |\bigcup_{r \geq p} \T(s_r)| > (|\S_j| - p + 1)(T - 1) + 1.
\end{align}
Since $|\T(s_i^r)| \leq T-1$ whenever $r \geq 2$, and $p$ is the largest index such that there is no $s_p, t$ dipath for some $t \in \T(s_p)$, it follows that $s_p = s_i^1$ for some~${i \in N_j}$. Thus, an equivalent statement to (\ref{ineq1}) is
\begin{align}\label{ineq2}
    \sum_{r \in N_j: r \geq i} |T_r| > 1 + \sum_{r \in N_j: r \geq i} |S_r| (T-1) .
\end{align}

Due to the ordering of the sets $S_1$ to $S_n$ in $D$ and by Claim 2, inequality (\ref{ineq2}) cannot be satisfied if $w_i \leq 0$. 
Alternatively, $w_i \geq 2$. By Claim 1 and the ordering of the source sets, 
\begin{align}\label{ineq3}
\sum_{r \in N_j: r < i} |T_r| \geq \sum_{r \in N_j: r < i} |S_r|T. 
\end{align}
But then by combining inequalities (\ref{ineq2}) and (\ref{ineq3}), we see that $|\T_j| > |\S_j| (T-1) + 1$, a contradiction. Therefore, $H_j$ is feasible as claimed, and $\Delta^* \leq T$.
\end{proof}